\documentclass[reqno]{amsart}
\usepackage[utf8]{inputenc}
\usepackage{amsmath,amsthm,amsfonts,amscd,amssymb,amstext}
\usepackage{soul}
\usepackage{mathtools}
\numberwithin{equation}{section}
\usepackage{booktabs}
\usepackage{esint}
\usepackage{relsize}
\usepackage{graphicx,color}
\usepackage{tabularx}
\usepackage{titletoc}
\usepackage{microtype}
\usepackage[colorlinks=true,linkcolor=black,citecolor=red,urlcolor=blue]{hyperref}
\usepackage{listings}
\usepackage{cite}
\usepackage{url}
\usepackage{alltt}
\usepackage{tikz}
\usetikzlibrary{arrows,cd,decorations.markings,positioning}
\usepackage{tkz-graph}
\tikzstyle{vecArrow} = [thick, decoration={markings,mark=at position
   1 with {\arrow[semithick]{open triangle 60}}},
   double distance=1.4pt, shorten >= 5.5pt,
   preaction = {decorate},
   postaction = {draw,line width=1.4pt, white,shorten >= 4.5pt}]
\tikzstyle{innerWhite} = [semithick, white,line width=1.4pt, shorten >= 4.5pt]

\theoremstyle{plain}
\newtheorem{theorem}{Theorem}

\newtheorem{lem}{Lemma}

\newtheorem{problem}{Problem}

\theoremstyle{remark}
\newtheorem*{rmk}{Remark}

\newcommand{\barr}{\begin{eqnarray}}
\newcommand{\earr}{\end{eqnarray}}
\newcommand{\be}{\begin{equation}}
\newcommand{\ee}{\end{equation}}

\newcommand{\de}{\mathrm{d}}

\newcommand{\ket}[1]{\left| #1 \right>} 
\newcommand{\bra}[1]{\left< #1 \right|} 
\let\baraccent=\= 
\renewcommand{\=}[1]{\stackrel{#1}{=}} 

\newcommand{\numberset}{\mathbb}

\newcommand{\Z}{\numberset{Z}}
\newcommand{\R}{\numberset{R}}
\newcommand{\C}{\numberset{C}}

\newcommand{\Tr}{\mathrm{Tr}}

\newcommand{\diag}{\mathrm{diag}}

\newcommand{\E}{\mathbb{E}}

\newcommand{\XX}{\mathcal{X}}

\newcommand{\Ai}{\mathrm{Ai}}

\def\K{\mathcal{K}}

\begin{document}

\title[Free fermions and the classical compact groups]{Free fermions and the classical compact groups}

\author[F. D. Cunden]{Fabio Deelan Cunden}
\address{F. D. Cunden, School of Mathematics, University of Bristol, University Walk, Bristol BS8 1TW, England, and School of Mathematics and Statistics, University College Dublin, Belfield, Dublin 4, Ireland}
\author[F. Mezzadri]{Francesco Mezzadri}
\address{F. Mezzadri, School of Mathematics, University of Bristol, University Walk, Bristol BS8 1TW, England}
\author[N. O'Connell]{Neil O'Connell}
\address{N. O'Connell, School of Mathematics, University of Bristol, University Walk, Bristol BS8 1TW, England, and School of Mathematics and Statistics, University College Dublin, Belfield, Dublin 4, Ireland}
\date{\today}

\begin{abstract} There is a close connection between the ground state of non-interacting  fermions in a box with classical (absorbing, reflecting, and periodic) boundary conditions and the eigenvalue statistics of the classical compact groups. The associated determinantal point processes can be extended in two natural directions: i) we consider the full family of admissible quantum boundary conditions (i.e., self-adjoint extensions) for the Laplacian on a bounded interval, and the corresponding projection correlation kernels; ii) we construct the grand canonical extensions at finite temperature of the projection kernels, interpolating from Poisson to random matrix eigenvalue statistics.
The scaling limits in the bulk and at the edges are studied in a unified framework, and the question of universality is addressed. Whether the finite temperature determinantal processes correspond to the eigenvalue statistics of some matrix models is, a priori, not obvious. We complete the picture by constructing a finite temperature extension of the Haar measure on the classical compact groups. The eigenvalue statistics of the resulting grand canonical matrix models (of random size) corresponds exactly to the grand canonical measure of free fermions with classical boundary conditions.
\end{abstract}

\maketitle
\tableofcontents
\section{Introduction}
\label{sec:intro}
In this paper we introduce and discuss several extensions of the eigenvalue statistics induced by the Haar measure on the classical compact groups $\mathrm{U}(2N+1)$, $\mathrm{Sp}(2N)$, $\mathrm{SO}(2N+1)$, and $\mathrm{SO}(2N)$. 

The starting point of this work is the following connection between the classical compact groups and free fermions in the ground state:

\smallskip
\emph{The eigenvalues of random matrices sampled according to the Haar measure on the classical compact groups, and the particle density of free (non-interacting) fermions in a box with classical boundary conditions at zero temperature, form the same determinantal point processes.}
\smallskip

This follows from well known formulae for the joint law of eigenvalues of random matrices, and elementary diagonalisation of Schr\"odinger operators. 
The cases $\mathrm{U}(2N+1)$, $\mathrm{Sp}(2N)$, and $\mathrm{SO}(2N)$ correspond to the most common textbook examples of `particles in a box', and have been pointed out and discussed in the literature (see, e.g.~\cite{Forrester03,Forrester11,Gangardt00}). Nevertheless, this mapping has not been appreciated enough and suggests two natural `extensions' of the determinantal processes associated to the classical compact groups.

First, we investigate the process associated to the ground state of non-interacting fermions in a box with \emph{generic quantum boundary conditions}. Recall that the physical dynamics of closed quantum system is a strongly continuous one-parameter unitary evolutions. By Stone's theorem, the generator of the unitary group, i.e. the Hamiltonian, must be a self-adjoint operator. See e.g.~\cite{Teschl14}. It is therefore legitimate to consider the whole family of self-adjoint extensions of the Laplacian on a bounded interval (kinetic energy in a box). In fact, the Laplacian on a bounded interval admits infinitely many self-adjoint extensions, each one characterised by the behaviour of the wavefunction at the boundary points.
By considering all the admissible boundary conditions, we show that the processes defined by the Haar measure on the classical compact groups are immersed in a four-parameter family of determinantal processes associated to free fermions in a box. The special cases of periodic, absorbing and reflecting boundary conditions correspond to the eigenvalue statistics of the classical groups. The choice of different self-adjoint extensions of the Laplacian is not just a mathematical nuisance. Different boundary conditions give rise to different physics, and their role and importance at a fundamental level has been recently stressed in a series of interesting articles, see~\cite{Asorey,Asorey13,DellAntonio00,Shapere12,Facchi16} and reference therein, where varying boundary conditions are viewed as a model of spacetime topology change.

A second natural extension consists in considering free fermions in a box at \emph{finite temperature}. 
These finite temperature extensions of the eigenvalue statistics of the classical compact groups are introduced with the purpose of providing a realistic statistical description of the transition between Poisson to random matrix eigenvalue statistics. 
This is not the first proposal of finite temperature extension of random matrix eigenvalue processes. There exists a well studied finite temperature extension of the celebrated GUE process. See e.g.~\cite{Dean16,Eisler13,Johansson07,Lambert15,LeDoussal17,Moshe94,Vicari12}. Nevertheless, the analogue for the eigenvalue statistics of the classical group is considerably more neat.
The finite temperature versions of the eigenvalue process of the classical groups have a (grand canonical) determinantal structure. Amusingly, they have the striking property of being the eigenvalue processes of random matrices (of random size), i.e., they describe the zeros of random characteristic polynomials (of random degree). These new ensembles of random matrices are constructed by i) `evolving' the Haar measure along the heat flow on the classical compact groups, and ii) by considering a suitable randomization on the size of the group (grand canonical construction).

\subsection{Eigenvalue statistics of random matrices} Let $X$ be a random $N \times N$ Hermitian matrix distributed according to the unitarily invariant measure
\be
P_N(X)\de X=C_N\exp(-2\Tr V(X))\de X.
\label{eq:matrixmodel}
\ee
Denote by $\phi_k$, $k=0,1,\dots$, the orthonormal polynomials ($\int\overline{\phi_k}(x)\phi_{\ell}(x)e^{-V(x)}\de x=\delta_{k\ell}$) with respect to the weight $e^{-V(x)}\de x$,  and consider the kernel 
\be
\Phi^V(x,y)=\sum_{k=0}^{N-1}\overline{\phi_k}(x)\phi_k(y)e^{-(V(x)+V(y))}.
\ee
It can be shown that the eigenvalues of $X$ form a determinantal point process with kernel $\Phi^V(x,y)$. In particular, their joint distribution is
\be
p_N(x_1,\dots,x_N)=\frac{1}{N!}\det[\Phi^V(x_i,x_j)]_{i,j=1}^N.
\label{eq:eig_dens}
\ee

\subsection{Ground state of non-interacting fermions} 
\label{sub:GS} Consider the ground state of $N$ non-interacting spin-polarized fermions in a trapping potential $V(x)$. In formulae we consider the many-body Schr\"odinger equation 
\be
\left[\sum_{i=1}^N-\frac{\partial^2}{\partial x_i^2}+V(x_i)\right]\varPsi(x_1,\dots,x_N)=E\varPsi(x_1,\dots,x_N),\label{eq:MBSE}
\ee
where $\varPsi$ denotes an antisymmetric normalised wavefunction ($\varPsi(x_{\pi(1)},\dots,x_{\pi(N)})=\mathrm{sgn}(\pi)\varPsi(x_1,\dots,x_N)$, and $\int |\varPsi(x_1,\dots,x_N)|^2\de x_1\cdots\de x_N=1$). At zero temperature, $N$ fermions are in the ground state (lowest energy state) given by the well-known Slater determinant formula. Therefore, the probability density $|\varPsi(x_1,\dots,x_N)|^2$ can be written as
\be
|\varPsi(x_1,\dots,x_N)|^2=\frac{1}{N!}\det[\Psi^V(x_i,x_j)]_{i,j=1}^N,
\label{eq:part_dens}
\ee
where 
\be
\Psi^V(x,y)=\sum_{k=0}^{N-1}\overline{\psi_k}(x)\psi_k(y),
\ee 
and the functions $\psi_k$ are the first $N$ eigenfunctions of the single-particle Schr\"odinger operator 
\be
-\psi''_k(x)+V(x) \psi_k(x)=E_k\psi_k(x),\label{eq:SPSE}
\ee
These eigenfunctions are orthonormal   $\int\overline{\psi_k}(x)\psi_{\ell}(x)\de x=\delta_{k\ell}$ and, therefore, $\Psi^V(x,y)$ defines a determinantal process.

\subsection{The GUE process} For a given potential $V(x)$, the eigenvalue process~\eqref{eq:eig_dens} of the matrix model~\eqref{eq:matrixmodel} and the particle density~\eqref{eq:part_dens} in the ground state of the Schr\"odinger operator~\eqref{eq:MBSE} are,  in general, unrelated. A notable exception is the case of a quadratic potential $V(x)=x^2/4$, when $\Phi^V(x,y)=\Psi^V(x,y)=K_{\mathrm{GUE}(N)}(x,y)$ is the kernel of the GUE ensemble of random matrix theory
\be
K_{\mathrm{GUE}(N)}(x,y)=\sum_{k=0}^{N-1}h_k(x)h_k(y)e^{-(x^2+y^2)/4},
\label{eq:kernelGUE}
\ee
where $h_k(x)$ are the rescaled Hermite polynomials
\be
h_k(x)=\frac{(-1)^k}{\sqrt{\sqrt{2\pi}k!}}e^{x^2/2}\frac{\de^k}{\de x^k}e^{-x^2/2}.
\ee 
The correlation kernel~\eqref{eq:kernelGUE} is that of the GUE eigenvalue process. This is the relation between non-interacting fermions in a harmonic potential at zero temperature and GUE matrices.

It can be shown that in some scalings (a change of variable depending on $N$), the GUE process converges as $N\to\infty$ to a  point process whose correlation functions are determined by the scaling limit of the kernel. More precisely, the GUE correlation kernel converges to the sine kernel (in the bulk) and to the Airy kernel (at the edge):
\begin{align}
&\hspace{-2mm}\frac{\pi}{\sqrt{N}}K_{\mathrm{GUE}(N)}\left(\frac{\pi x}{\sqrt{N}},\frac{\pi y}{\sqrt{N}}\right)\stackrel{N\to\infty}{\longrightarrow}\frac{\sin(\pi(x-y))}{\pi(x-y)},\\
&\hspace{-2mm}\frac{1}{N^{\frac{1}{6}}}K_{\mathrm{GUE}(N)}\left(2\sqrt{N}+\frac{x}{N^{\frac{1}{6}}},2\sqrt{N}+\frac{ y}{N^{\frac{1}{6}}}\right)\stackrel{N\to\infty}{\longrightarrow}\frac{\Ai(x)\Ai'(y)-\Ai'(x)\Ai(y)}{x-y}.
\end{align}

\begin{problem}[Mappings between matrix ensembles and non-interacting fermions]\label{prob:1} Discuss other examples of exact correspondence between complex random matrices and the ground state of Schr\"odinger operators on non-interacting fermions. In formulae, we look for a potential $V(x)$ such that the kernel of the eigenvalue process is identical to the kernel of the fermions density, $\Phi^V(x,y)=\Psi^V(x,y)$. (Note that in general, for a given potential $V(x)$, different boundary conditions correspond to different Schr\"odinger operators.) For those examples, discuss the scaling limits and address the question of their universality.
\end{problem}

\subsection{Finite temperature GUE} One can push further the correspondence for GUE as follows. The solutions of the single-particle Schr\"odinger equation~\eqref{eq:SPSE} with quadratic potential $V(x)=x^2/4$ are $\psi_k(x)=h_k(x)e^{-x^2/4}$ and $E_k=k+1/2$ ($k=0,1,2,\dots$). One then defines the \emph{finite temperature} $\mathrm{GUE}$ process as the grand canonical process with correlation kernel 
\be
K_{\mathrm{GUE}(T,\mu)}(x,y)=\sum_{k=0}^{\infty}\frac{\psi_{T,k}(x)\psi_{T,k}(y)}{1+e^{-(\mu-k-1/2)/T}},
\label{eq:GUET}
\ee
where $\psi_{T,k}(x)=\sqrt[4]{\coth(1/2T)}\;\psi_k(\sqrt{\coth(1/2T)}x)$ are rescaled wavefunctions, and the chemical potential $\mu=\mu(N,T)$  is fixed by the condition
\be
N=\sum_{k=0}^{\infty}\frac{1}{1+e^{-(\mu-k-1/2)/T}}.
\label{eq:GUEconstr}
\ee
The kernel~\eqref{eq:GUET} defines the grand canonical measure of a system of non-interacting fermions in a harmonic potential at temperature $T>0$ and chemical potential $\mu>0$ (such that the average number of fermions is $N$).
Johansson~\cite{Johansson07} proved that such a grand canonical process interpolates between a point process defined by $N$ independent Gaussian and eigenvalues of GUE matrices, as expected.
 Moreover, in a suitable rescaling of the temperature with the number of particles, one obtains a family of limiting kernels that extends the classical sine kernel and Airy kernel of random matrix theory:

\begin{itemize}
\item[i)] (Interpolation between Poisson and GUE.)
\be
\lim_{T\to 0}K_{\mathrm{GUE}(T,\mu)}(x,y)=K_{\mathrm{GUE}(N)}(x,y)
\label{eq:Johansson1}
\ee
uniformly for $x,y$ in a compact set, and
\be
\lim_{T\to\infty}K_{\mathrm{GUE}(T,\mu)}(x,y)=
\left\{  \begin{array}{l@{\quad}cr} 
0&\text{if $x\neq y$,}\\
\frac{N}{\sqrt{\pi}}e^{-x^2}&\text{if $x=y$,}
\end{array}\right.
\label{eq:Johansson2}
\ee
pointwise;
\item[ii)] (Limit of high temperature and large number of particles in the bulk.)

\noindent Let $T=cN$, with $c>0$ fixed, and $\mu=cN\log{\lambda}$ with\footnote{$\operatorname{Li}_s(z)$ is the polylogarithm function. It is the analytic extension of the Dirichlet series $\sum_{k=1}^{\infty}\frac{z^k}{k^s}$. } $\lambda=-\operatorname{Li}_{1}^{-1}(-1/c)=e^{1/c}-1$. 
The following limit holds
\be
\frac{\pi}{N\sqrt{c}}K_{\mathrm{GUE}(cN,cN\log\lambda)}\left(\frac{\pi x}{N\sqrt{c}},\frac{\pi y}{N\sqrt{c}}\right)\stackrel{N\to\infty}{\longrightarrow}\int_0^{\infty}\frac{\cos{(\pi(x-y)u)}}{1+\lambda^{-1}e^{u^2/c}}\de u,
\label{eq:Johansson3}
\ee
uniformly for $x,y$ in a compact set;
\item[iii)] (Limit of high temperature and large number of particles at the edge.) 

\noindent Let $T=cN^{1/3}$, and $e^{\frac{\mu}{T}}=e^{\frac{1}{c}}-1$, where $c>0$ is fixed. Then,
\begin{align}
\frac{1}{N^{\frac{1}{3}}\sqrt{c}}K_{\mathrm{GUE}(cN^{\frac{1}{3}},\mu)}\left(N^{\frac{1}{3}}\sqrt{c}+\frac{x}{N^{\frac{1}{3}}\sqrt{c}},N^{\frac{1}{3}}\sqrt{c}+\frac{y}{N^{\frac{1}{3}}\sqrt{c}}\right)\nonumber\\
\stackrel{N\to\infty}{\longrightarrow}\int_{-\infty}^{\infty}\frac{\mathrm{Ai}(x+u)\mathrm{Ai}(y+u)}{1+e^{-u/c}}\de u,
\end{align}
uniformly for $x,y$ in a compact set.
\end{itemize}
The finite temperature GUE model and the associated limit kernels have been studied in several papers.  See~\cite{Dean16,Eisler13,Johansson07,Kohn98,Lambert15,Lenard66,Liechty17,Moshe94,Vicari12}.
\begin{problem}[Extensions from ground state to finite temperature]\label{prob:2} For the new examples of Problem~\ref{prob:1}, construct the finite temperature extensions, show that these ensembles interpolate between random matrix and Poisson statistics, and compute the nontrivial scaling limits. Address the question of the universality of the limiting kernels.
\end{problem}
\subsection{The grand canonical MNS ensemble} A natural question is whether the finite temperature GUE process corresponds, in some sense, to the eigenvalue process of a matrix model. Of course, this cannot be strictly true, since the number of points $N$ in $\mathrm{GUE}(T,\mu)$ is not fixed. It turns out that the $\mathrm{GUE}(T,\mu)$ process describes the statistics of an ensemble of random Hermitian matrices whose size $N$ is itself a random variable.

The MNS model of $n\times n$ Hermitian matrices is a unitarily invariant ensemble defined by the probability measure
\be
P_{n,t}(X)\de X=C_{n,t}e^{-\frac{1}{2}\Tr X^2}\left(\int_{\mathrm{U}(n)}\exp({-\frac{1}{2t}\Tr([V,X][V,X]^{\dagger}})\de V\right)\de X.
\label{eq:MSNmodel}
\ee
This ensemble has been invented by Moshe, Neuberger and  Shapiro~\cite{Moshe94}. They showed that the joint distribution of the eigenvalues of $X$ is 
\barr
p_{n,t}(x_1,\dots,x_n)=\frac{1}{Z_{n}}\det\left[\frac{1}{(2\pi t)^{\frac{1}{2n}}}e^{-\frac{1}{4}(x_i^2+x_j^2)}e^{-\frac{1}{2t}(x_i-x_j)^2}\right]_{i,j=1}^n,
\label{eq:MSN_jpdf}
\earr
where $Z_n$ is the normalisation constant (depending on $t$).
Setting $t=2\sinh^2(1/2T)$, the function inside the determinant is the so-called \emph{canonical kernel}
\be
\frac{1}{\sqrt{2\pi t}}e^{-\frac{1}{4}(x^2+y^2)}e^{-\frac{1}{2t}(x-y)^2}=e^{-\frac{1}{4}(x^2+y^2)}\sum_{k=0}^{\infty}e^{-(k+1/2)/T}h_k(x)h_k(y).
\ee

The eigenvalues of the MNS model do not form a determinantal point process. One can construct the grand canonical point process by considering a MNS measure on matrices of size $N$ and letting $N$ be an integer valued random variable with
\be
\Pr(N=n)=\frac{1}{Z(\mu)}\exp\left({\frac{\mu}{T}n}\right)\frac{Z_n}{n!},\quad Z(\mu)=\sum_{n=0}^{\infty}\exp\left({\frac{\mu}{T}n}\right)\frac{Z_n}{n!}, \quad (\mu>0).
\ee
This grand canonical MNS model is an ensemble of random matrices of random size $N$; given $N=n$, the joint distribution of the eigenvalues is~\eqref{eq:MSN_jpdf}. One can show (see~\cite{Johansson07})  that the eigenvalues of this ensemble form a determinantal point process whose kernel is $K_{\mathrm{GUE}(T,\mu)}(x,y)$. Hence, the grand canonical version of the MNS model provides a matrix realisation of the finite temperature GUE process.
\begin{problem}[Back to random matrices]\label{prob:3} 
Construct a (grand canonical) random matrix model whose eigenvalue statistics is one of the finite temperature processes of Problem~\ref{prob:2}.
\end{problem}

The rest of this paper is organised as follows:
\begin{itemize}
\item[(i)] In Section~\ref{sec:DPP} and Section~\ref{sec:CCG} we collect some basic facts about determinantal point processes and   the eigenvalues statistics induced by the Haar measure on the classical compact groups.
\item[(ii)] In Section~\ref{sec:bc} we provide an answer to Problem~\ref{prob:1}. We discuss the precise correspondence between classical compact groups and free fermions confined in an box (or, equivalently, fermions on a circle with a zero-range perturbation at a fixed point). Each group corresponds to a particular self-adjoint extension (i.e. boundary conditions) of $-\Delta$ on $(0,2\pi)$. 
\item[(iii)] In Section~\ref{sec:saext}  we extend the kernels of the classical compact groups by considering the whole family of self-adjoint extensions of $-\Delta$ on $(0,2\pi)$. For these determinantal processes we study the scaling limit on the scale of the mean level spacing of the particles. In the bulk, we prove the universality of the sine kernel. At the edges $0$ and $2\pi$, the limiting process depends on the quantum boundary conditions. Absorbing and reflecting boundary conditions correspond to Bessel processes. Elastic (Robin) boundary conditions and $\delta$-perturbations lead to new one-parameter kernels. 
\item[(iv)] In Section~\ref{sec:CUET} we address Problem~\ref{prob:2} and we propose a finite temperature extension of the eigenvalues statistics of the classical compact groups. We show that these determinantal processes interpolate between random matrix and Poisson statistics and we investigate the simultaneous limit of high temperature and large number of particles. In the bulk the limit process is the same finite temperature sine process emerging in the finite temperature GUE. 
\item[(v)] In Section~\ref{sec:MM} we provide a systematic answer to  Problem~\ref{prob:3}. We first show that the MNS model is related to a matrix integral of the heat kernel $k_t$ on the algebra of Hermitian matrices. This remark suggests to extend this construction to Lie  groups by using the group heat kernel $K_t$. It turns out that this construction provides an analogue of the MSN model for the classical compact groups. The grand canonical version of these new ensembles forms exactly the finite temperature determinantal processes constructed in Section~\ref{sec:CUET}. 
\end{itemize}
\section{Determinantal point processes}
\label{sec:DPP}
A \emph{point process} (or random point field) on a locally compact space $\XX$ equipped with some reference measure $\de\mu$ is a random measure on $\XX$ of the form $\sum_{i}\delta_{X_i}$. The support of the measure can be finite or countably infinite, but it cannot have accumulation points in $\XX$. Point processes are usually described by their correlation functions $\rho_n(x_1,\dots,x_n)$ defined by the formula
\be
\E\prod_{i=1}(1+g(X_i))=\sum_{n=0}^{\infty}\frac{1}{n!}\int\limits_{\XX^n}\rho_n(x_1,\dots,x_n)\prod_{i=1}^{n}g(x_i)\de\mu(x_i)
\ee
for any measurable functions $g\colon\XX\to\C$ with compact support. A point process is called \emph{determinantal} if its correlation functions exist and satisfy the identity
\be
\rho_n(x_1,\dots,x_n)=\det[K(x_i,x_j)]_{i,j=1}^n,
\ee
where the \emph{correlation kernel} $K\colon\XX\times\XX\to\C$ is independent on $n$. The correlation kernel is not unique: replacing  $K(x,y)$ by $f(x)K(x,y)f(y)^{-1}$, where $f$ is an arbitrary nonzero function, leaves the determinants $\det [K(x_i,x_j)]$ intact.

It is useful to view the function $K(x,y)$ as the kernel of an integral operator $\K$ acting in the Hilbert space $L^{2}(\XX,\mu)$. Assume that $\K$  is self-adjoint and locally of trace class. Then, $K(x,y)$ is the correlation kernel of a determinantal point process if and only if the operator $\K$ satisfies the condition $0\leq \K\leq I$. In such a case, the kernel can be written generically as  
\be
K(x,y)=\sum_kp_k\overline{\psi_k}(x)\psi_k(y),\label{eq:K}
\ee
where $(\psi_k)$ is an orthonormal basis in $L^2(\XX,\mu)$ and $0\leq p_k\leq1$. In this paper we shall often use the (Dirac) notation $K(x,y)=\bra{x}\K\ket{y}$. 

We will focus on the following two classes:
\begin{enumerate}
\item\emph{Zero temperature processes} whose kernels have the form~\eqref{eq:K} with
\be
\text{$p_1=\cdots =p_N=1$ and $p_k=0$ for $k>N$},
\ee  
for some finite $N$. In this case, $\K$ is a $N$-dimensional orthogonal projection operator. The number of particles in a zero temperature process is $N$ almost surely.
\item\emph{Grand canonical processes}~\cite{Johansson07} whose kernel has the form~\eqref{eq:K} with
\be
p_k=\frac{1}{1+e^{-(\mu-E_k)/T} },
\ee  
where $\mu,T>0$ and $\sum_ke^{-E_k/T}<\infty$. The number of particles $N$ in a grand canonical process is not fixed ($N$ fluctuates).
\end{enumerate}
A Poisson process on $\XX$ with density $\rho(x)$ can be viewed as a, somewhat degenerate, determinantal process with correlation kernel
\be
K(x,y)=
\left\{  \begin{array}{l@{\quad}cr} 
0&\text{if $x\neq y$,}\\
\rho(x)&\text{if $x=y$.}
\end{array}\right.
\ee 
For more details on determinantal random point fields, see~\cite{Johansson05,Peres06,Soshnikov00}.

\section{Haar measure on the classical compact groups}
\label{sec:CCG}
We introduce the notation
\be
S_N(z)=
\left\{  \begin{array}{l@{\quad}cr} 
\displaystyle\frac{1}{2\pi}\frac{\sin(Nz/2)}{\sin(z/2)}&\text{if $z\neq 0$,}\\
\displaystyle\frac{N}{2\pi}&\text{if $z=0$.}
\end{array}\right.
\ee
Let $U$ be a random matrix distributed according to the normalized Haar measure on $\mathrm{U}(N)$ (the so-called circular unitary ensemble (CUE) in random matrix theory). The eigenvalues of $U$ have joint density
\be
P_{\mathrm{U}(N)}(x_1,\dots,x_N)=\frac{1}{N!(2\pi)^N}\prod_{j<k}|e^{ ix_j}-e^{ ix_k}|^2
\ee
with respect to $\de x_1\cdots\de x_N$ on $[0,2\pi)^N$. 

Consider a matrix $U$ distributed according to the normalized Haar measure on $G$, where $G$ is one of the groups $\mathrm{Sp}(2N)$, $\mathrm{SO}(2N)$, $\mathrm{SO}(2N+1)$. Note that each matrix in $\mathrm{SO}(2N+1)$ has $1$ as eigenvalue; we refer to this as trivial eigenvalue. The remaining eigenvalues of matrices in $G$ occur in complex conjugate. Then, the $N$ nontrivial eigenvalues of $U$ in the open upper half-plane have joint density with respect to $\de x_1\cdots\de x_N$ on $[0,\pi)^N$ given by
\begin{align}
P_{\mathrm{Sp}(2N)}(x_1,\dots,x_N)&=\displaystyle\frac{2^N}{N!(\pi)^N}\prod_j\sin^2(x_j)\prod_{j<k}(2\cos{x_j}-2\cos{x_k})^2,\\
P_{\mathrm{SO}(2N)}(x_1,\dots,x_N)&=\displaystyle\frac{2}{N!(2\pi)^N}\prod_{j<k}(2\cos{x_j}-2\cos{x_k})^2,\\
P_{\mathrm{SO}(2N+1)}(x_1,\dots,x_N)&=\displaystyle\frac{2^N}{N!(\pi)^N}\prod_j\sin^2(x_j/2)\prod_{j<k}(2\cos{x_j}-2\cos{x_k})^2.
\end{align} 
Moreover, the nontrivial eigenvalue angles of a random $U$ form a determinantal process in $\Lambda$ (i.e., $P_G(x_1,\dots,x_N)=(N!)^{-1}\det [Q_G(x_i,x_j)]_{i,j=1}^N$) with correlation kernels
\begin{align}
Q_{\mathrm{U}(N)}(x,y)&=S_N(x-y),\\
Q_{\mathrm{Sp}(2N)}(x,y)&=S_{2N+1}(x-y)-S_{2N+1}(x+y),\\
Q_{\mathrm{SO}(2N)}(x,y)&=S_{2N-1}(x-y)+S_{2N-1}(x+y),\\
Q_{\mathrm{SO}(2N+1)}(x,y)&=S_{2N}(x-y)-S_{2N}(x+y),
\end{align}
where $\Lambda=[0,2\pi)$ in the first case, and $\Lambda=[0,\pi)$ otherwise. 
In the bulk of the spectrum, the sine process describes the eigenvalue distribution of random
matrices on the scale of the mean eigenvalue spacing
\be
\lim_{N\to\infty}\frac{2\pi}{N}Q_{\mathrm{U}(N)}\left(x_0+\frac{2\pi x}{N},x_0+\frac{2\pi y}{N}\right)=\frac{\sin(\pi(x-y))}{\pi(x-y)},\quad \text{for all }x_0\in[0,2\pi),
\label{eq:scalingUN}
\ee
\be
\lim_{N\to\infty}\frac{\pi}{N}Q_{G}\left(x_0+\frac{\pi x}{N},x_0+\frac{\pi  y}{N}\right)=\frac{\sin(\pi(x-y))}{\pi(x-y)},\quad \text{for all }x_0\in(0,\pi),
\label{eq:scalingG}
\ee
where $G=\mathrm{Sp}(2N)$, $\mathrm{SO}(2N)$, and $\mathrm{SO}(2N+1)$.

\section{Non-interacting fermions in a box and the classical compact groups}
\label{sec:bc}
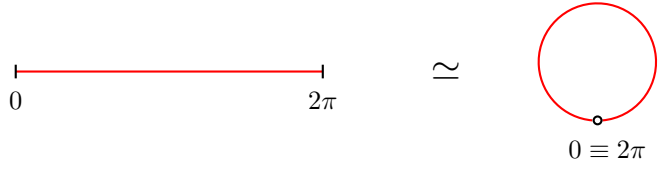
\begin{figure}[t]
\centering
\begin{tikzpicture}
[scale=1.3
]
\draw[red,thick]
  ([shift={(0:0cm)}]4,0) arc (-90+5:270-5:.6cm);
\draw[thick] (3.95,0) circle (.04);

\draw[red,thick] (-2cm,.5cm) -- (1.14cm,.5cm);
\draw [ thick] (1.14,0.57) -- (1.14,0.43);
\draw [ thick] (-2,0.57) -- (-2,0.43);

\node at (2.4,.5) {{\LARGE $\simeq$}};
\node at (-2,.5-0.3) {$0$};
\node at (1.14,.5-0.3) {$2\pi$};
\node at (4.05,-0.3) {$0\equiv 2\pi$};
\end{tikzpicture} 
\caption{After bending, the interval $[0,2\pi)$ transforms as in figure in the unit circle. The boundary conditions of functions at the edges of the interval become boundary conditions on the left and and the right of the junction at $0\equiv2\pi$. }
\label{fig:interval-circle}
\end{figure}
In this section we present new and interesting examples where there exists a precise correspondence between non-interacting fermions and matrix models. The differential operator 
\be
H\psi(x)=-\psi''(x),\quad\psi\in C^{\infty}_0(0,2\pi)
\label{eq:Hform}
\ee
is a (closable) symmetric operator, the self-adjoint extensions of which are considered as realisations of a `particle in a box'. Equivalently, the self-adjoint extensions of $H$ are considered as `perturbations' of the  Laplacian on the unit circle by a zero-range (singular) potential supported at point $0$ identified with the point $2\pi$ (see Figure~\ref{fig:interval-circle}). 

The self-adjoint extensions $H_U$ of $H$ are labelled bijectively by elements of the group $\mathrm{U}(m)$ where $m$ is the deficiency index of $H$~\cite{vonNeumann,Teschl14}. Moreover, it is a classical result~\cite{Naimark68} that, for a differential operator of order $m$ with deficiency index $m$, all of its self-adjoint extensions have only discrete spectrum. It is a simple exercise to show that, for the operator~\eqref{eq:Hform}, $m=2$ and hence $H_U$, defined on $D(H_U)$, can be parametrized by the set of $2\times2$ unitary matrices. Altogether there are four independent real coordinates to parametrize the set of self-adjoint extensions of the Laplacian on a finite interval, as $\dim_{\R}\mathrm{U}(2)=4$, and the meaning of the parameters is that they fix the boundary conditions (b.c.).

Let us consider $N$ non-interacting spin-polarized, or spinless, fermions confined in the box of length $2\pi$. If we fix the boundary conditions, the ground state is the Slater determinant of the first $N$ eigenfunctions of the single-particle Schr\"odinger operator, that is the solutions $\psi_{E_k}$ of
\be
H_U\psi_{E_k}(x)=E_k\psi_k(x),\quad \psi_k\in D(H_U).
\ee

We first focus on the classical boundary conditions, periodic (P), Dirichlet (D), Neumann (N), and Zaremba (Z), corresponding to four self-adjoint extensions of $H$. The ground state particle density of the free fermions forms a determinantal process whose correlation kernel is the kernel of the spectral projection onto the first $N$ single-particle eigenfunctions (see Section~\ref{sub:GS}). In the following, we show that, in the case of the classical boundary conditions, the point processes are the same as the eigenvalue processes induced by the Haar measure on the classical groups $G=\mathrm{U}(2N+1)$, $\mathrm{Sp}(2N)$, $\mathrm{SO}(2N)$, and $\mathrm{SO}(2N+1)$. This exact correspondence provides an answer to Problem~\ref{prob:1} by formally considering the potential $V(x)=0$ for $x\in(0,2\pi)$, and $+\infty$ for $x\notin(0,2\pi)$, often denoted as `infinite potential well'. By imposing the specific behaviour of the wavefunctions at the edges $0$ and $2\pi$ (i.e., the boundary conditions) we select among the classical groups. This correspondence is outlined below.

\subsection{Dirichlet b.c. and $\mathrm{Sp}(2N)$} For notational convenience, it is useful to identify functions $f(x)$ on $(0,2\pi)$ with functions $f(e^{ix})$ on the unit circle $S^1$. The limit values of $f(x)$ as $x$ goes to $0$ and $2\pi$, are then denoted simply as $f(0^{\pm})$.

 Consider the equation
\be
-\psi_{E_k}''(x)=E_k\psi_{E_k}(x),\quad x\in(0,2\pi)
\ee
with boundary conditions  $\psi_{E_k}(0^-)=\psi_{E_k}(0^+)=0$. A simple computation gives
\be
\psi_{E_k}(x)=\frac{1}{\sqrt{\pi}}\sin\left(\frac{kx}{2}\right),\quad E_k=\frac{k^2}{4},\quad k=1,2,\dots.
\ee
Therefore, see Section~\ref{sub:GS}, the particle density of $N$ free non-interacting fermions with Dirichlet b.c. is a determinantal point process with correlation kernel  
\begin{align}
K^{D}(x,y)&=\sum_{k=1}^N\overline{\psi_{E_k}}(x)\psi_{E_k}(y)\nonumber\\
&=\frac{1}{2\pi}\sum_{|k|\leq N}\sin\left(\frac{kx}{2}\right)\sin\left(\frac{ky}{2}\right)=\frac{1}{2}Q_{\mathrm{Sp}(2N)}\left(\frac{x}{2},\frac{y}{2}\right),
\end{align}
where $Q_{\mathrm{Sp}(2N)}$ is the rescaled correlation kernel of the Haar measure on the symplectic group $\mathrm{Sp}(2N)$. 

\subsection{Neumann b.c. and $\mathrm{SO}(2N)$}
The eigenfunctions $\psi_{E_k}$ and eigenvalues $E_k$ of the Schr\"odinger operator with Neumann b.c. $\psi_{E_k}'(0^-)=\psi_{E_k}'(0^+)=0$, are
\be
\psi_{E_0}(x)=\frac{1}{\sqrt{2\pi}},\,\, E_0=0,\quad
\psi_{E_k}(x)=\frac{1}{\sqrt{\pi}}\cos\left(\frac{kx}{2}\right),\,\, E_k= \frac{k^2}{4},\quad k=1,2,\dots.
\ee
A simple computation gives the correlation kernel of free fermions with Neumann b.c.
\begin{align}
K^{N}(x,y)&=\sum_{k=0}^{N-1}\overline{\psi_{E_k}}(x)\psi_{E_k}(y)\nonumber\\
&=\frac{1}{2\pi}\sum_{|k|\leq N-1}\cos\left(\frac{kx}{2}\right)\cos\left(\frac{ky}{2}\right)=\frac{1}{2}Q_{ \mathrm{SO}(2N)}\left(\frac{x}{2},\frac{y}{2}\right),
\end{align}
where $Q_{SO(2N)}(x,y)$ is the kernel of the Haar measure on the group  $\mathrm{SO}(2N)$ of special orthogonal matrices.

\subsection{Zaremba b.c. and $\mathrm{SO}(2N+1)$}
Let us consider the Zaremba (mixed) b.c.: one boundary condition is Dirichlet,  $\psi_{E_k}(0^-)=0$, and the other is Neumann $\psi_{E_k}'(0^+)=0$.
The eigenfunctions and eigenvalues of the Schr\"odinger operator are
\barr
\psi_{E_k}(x)=\frac{1}{\sqrt{\pi}}\sin\left(\frac{2k+1}{4}x\right),\,\, E_k= \left(\frac{2k+1}{4}\right)^2,\quad k=0,1,2,\dots.
\earr
Therefore, in this case,
\barr
K^{Z}(x,y)=\sum_{k=0}^{N-1}\overline{\psi_{E_k}}(x)\psi_{E_k}(y)=\frac{1}{2}Q_{\mathrm{SO}(2N+1)}\left(\frac{x}{2},\frac{y}{2}\right),
\earr
which is the rescaled kernel of the Haar measure on $\mathrm{SO}(2N+1)$.

\subsection{Periodic b.c. and $\mathrm{U}(2N+1)$}
Consider now the case of periodic boundary conditions $\psi_{E_k}(0^-)=\psi_{E_k}(0^+)$, and $\psi_{E_k}'(0^-)=\psi_{E_k}'(0^+)$. Note that the periodicity is a \emph{nonlocal} b.c. (it is useful to have in mind the picture in Figure~\ref{fig:interval-circle}). It is straightforward to solve the Schr\"odinger equation and find eigenfunctions $\psi_{E_k}(x)$ and eigenvalues $E_k$,
\be
\psi_{E_k}(x)=\frac{e^{ikx}}{\sqrt{2\pi}},\quad E_k=k^2,\quad k\in\Z.
\ee
Note that $E_k$ is doubly degenerate for $k\neq0$. Hence, the ground state of non-interacting fermions is non degenerate only in the case of odd number of particles. 
When considering $(2N+1)$ fermions at zero temperature we are led to consider the kernel
\be
K^{P}(x,y)=\sum_{|k|\leq{N}}\overline{\psi_{E_k}}(x)\psi_{E_k}(y)=\frac{1}{2\pi}\sum_{|k|\leq{N}}e^{ik(y-x)}= Q_{\mathrm{U}(2N+1)}(x,y),
\label{eq:kernel_P}
\ee
which is nothing but the correlation kernel of $\mathrm{U}(2N+1)$, that is the eigenvalues correlation kernel of a random unitary matrix of size $(2N+1)$ from the CUE. For pseudo-periodic b.c., that is $\psi_{E_k}(0^-)=e^{i\alpha}\psi_{E_k}(0^+)$, and $\psi_{E_k}'(0^-)=e^{i\alpha}\psi_{E_k}'(0^+)$ with $\alpha\in(0,2\pi)$, one obtains a kernel equivalent to that of CUE process.

At microscopic scale, the CUE process converges to a translation invariant process whose correlations are given by the sine kernel. Note that for Dirichlet, Neumann, and Zaremba conditions, the process is not translation invariant; nevertheless, in the `bulk', the scaling limit is again the sine process. 

We mention that particle fluctuations and entanglement measures of free fermions (with periodic or Dirichlet b.c.) have been recently studied in the physics literature by Calabrese, Mintchev and Vicari~\cite{Calabrese12}. High-dimensional generalisations of the kernel~\eqref{eq:kernel_P} (Fermi-shell models) have been proposed and investigated by  Torquato, Scardicchio and Zachary~\cite{Torquato08}. Forrester, Majumdar and Schehr studied at length the kernels $K^{D}$, $K^{N}$, and $K^{P}$, in the context of non-intersecting Brownian walkers and two-dimensional continuum Yang–Mills theory on the sphere~\cite{Forrester11}. 

Rescaling the kernels $K^{D}$, $K^{N}$, and $K^{Z}$ at the edge $0$, does not lead to the sine kernel. In fact, for Dirichlet and Neumann boundary conditions we obtain
\begin{align}
\displaystyle\frac{2\pi}{N}K^{D}\left(\frac{2\pi x}{N},\frac{2\pi y}{N}\right)&\stackrel{N\to\infty}{\longrightarrow} \frac{\sin(\pi(x-y))}{\pi(x-y)}-\frac{\sin(\pi(x+y))}{\pi(x+y)},\label{eq:D_Bessel}\\
\displaystyle\frac{2\pi}{N}K^{N}\left(\frac{2\pi x}{N},\frac{2\pi y}{N}\right)&\stackrel{N\to\infty}{\longrightarrow} \frac{\sin(\pi(x-y))}{\pi(x-y)}+\frac{\sin(\pi(x+y))}{\pi(x+y)}.\label{eq:N_Bessel}
\end{align}
These kernels and their Fredholm determinants have been studied in details in the early work by Dyson on real symmetric random matrices~\cite{Dyson76}, and more recently  by Katz and Sarnak to model the lowest zeros in families of L-functions~\cite{Katz99} (see also~\cite{Conrey,Keating}). They are related to special instances of the Bessel kernels
\be
B_{\nu}(x,y)=\frac{\sqrt{x}J_{\nu+1}(\sqrt{x})J_{\nu}(\sqrt{y})-J_{\nu}(\sqrt{x})\sqrt{y}J_{\nu+1}(\sqrt{y})}{2(x-y)},
\ee
where $J_{\nu}(x)$ is the ordinary Bessel function. A simple rescaling gives, for $\nu=\pm1/2$,
\be
2\pi^2\sqrt{xy}B_{\pm1/2}(\pi^2x^2,\pi^2y^2)=\frac{\sin(\pi(x-y))}{\pi(x-y)}\mp\frac{\sin(\pi(x+y))}{\pi(x+y)}.
\ee
When $\nu$ is an integer, the kernel $B_{\nu}(x,y)$ appears in the scaling limit around the smallest eigenvalue in  the Laguerre Unitary Ensemble of random matrices.
\section{Quantum boundary conditions and self-adjoint extensions}
\label{sec:saext} 
All the self-adjoint extensions of $H$, defined in~\eqref{eq:Hform}, are given by
\barr
D(H_U)&=&\left\{\psi\in H^2\left(0,2\pi\right)\colon 
\left(
\begin{array}{c}
\psi_-+i\psi'_-\\
\psi_+-i\psi'_+\\
\end{array}
\right)
=
U
\left(
\begin{array}{c}
\psi_--i\psi'_-\\
\psi_++i\psi'_+\\
\end{array}
\right)
\right\}
\\
H_U\psi(x)&=&-\psi''(x),\quad \psi\in D(H_U),
\earr
where $H^2\left(0,2\pi\right)$ is the second Sobolev space. $U\in\mathrm{U}(2)$ is a unitary matrix, $\psi_-=\psi(0^-)$, $\psi_+=\psi(0^+)$, $\psi'_-=\psi'(0^-)$ and $\psi'_+=\psi'(0^+)$. 
This parametrisation of the self-adjoint extension in terms of unitary operators on the boundary data, has been proposed on physical ground by Asorey, Marmo and Ibort~\cite{Asorey}, and has been applied to several one dimensional quantum systems (see, for instance,~\cite{Asorey13,Facchi16}). The self-adjoint operators $H_U$ correspond to a free particle in a box of length $2\pi$, or on the unit circle with a point perturbation\footnote{For periodic boundary conditions the point perturbation has strength zero.} at $0$. The choice of particular unitary matrices gives rise to some well-known boundary 
conditions, for example,
\begin{center}
\begin{tabular}{l|ll}
&\multicolumn{2}{c}{Boundary conditions}\\
 $U\in\mathrm{U}(2)$&  &  \\\hline\hline\\
$\sigma_1$ & Periodic &  $\psi_+=\psi_-,\quad\,\,\,\,\,\psi_+'=\psi_-'$ \\   
$\cos\alpha\sigma_1+\sin\alpha\sigma_2$ & Pseudo-periodic  & $\psi_+=e^{i\alpha}\psi_-,\,\,\,\psi_+'=e^{i\alpha}\psi_-'$  \\   
$-I$ & Dirichlet &  $\psi_+=\psi_-=0$ \\   
$I$ & Neumann & $\psi_+'=\psi_-'=0$  \\   
$-\sigma_3$ & Zaremba &  $\psi_-=0,\,\psi_+'=0$ \\   
$e^{i\alpha}I$ & Robin  & $\psi_\pm'=\pm\tan(\alpha/2)\psi_\pm$  \\   
$\frac{1}{1-ic/2}(\sigma_1-ic I/2)\,\,$ & $\delta$-potential ($-\Delta+c\delta$) & $\psi_+=\psi_-,\,\psi_+'-\psi_-'=c\psi_+$ \\   \smallskip
\end{tabular}
\end{center}
where $\sigma_1,\sigma_2,\sigma_3$ denote the $2\times2$ Pauli matrices.

Note that the Dirichlet, Neumann, Zaremba, and periodic b.c. correspond to four (out of an infinite family) self-adjoint extensions of the Laplacian. It is legitimate to investigate other boundary conditions. Consider, for instance, the Schr\"odinger operator  $H_{e^{i\alpha}I}$ 
corresponding to Robin boundary conditions. The eigenvalues $E_k$ are given by the solutions of a transcendental equation and, in general, the eigenfunctions $\psi_{E_k}$ are \emph{not} trigonometric polynomials. Nevertheless, one again expects the convergence to the sine process in the bulk (see below).
On the other hand, it is clear that the limiting behaviour at the edges depends on the boundary conditions, and is not universal. 

\subsection{Microscopic universality in the bulk}
The scaling transition to the sine process~\eqref{eq:scalingUN}-\eqref{eq:scalingG} for the classical b.c. can be written in a unified fashion as
\be
\lim_{E\to\infty}\frac{2\pi}{N(E)}\sum_{E_k\leq E}\overline{\psi_{E_k}}\left(x_0+\frac{2\pi x}{N(E)}\right)\psi_{E_k}\left(x_0+\frac{2\pi y}{N(E)}\right) =\frac{\sin(\pi(x-y))}{\pi(x-y)},\label{eq:unified_sink}
\ee
where $N(E)=\#\{E_k\leq E\}$ is the integrated density of states and $x_0\in(0,2\pi)$. 

In fact, we can ask whether the sine kernel is the universal limit in the bulk for \emph{all} self-adjoint extensions of the Laplacian.  To prepare the ground, it is useful to identify the sine kernel as the integral kernel of the kinetic energy operator of a free particle on the real line. Recall (see~\cite[Theorem 7.17]{Teschl14}) that the operator $-\partial^2/\partial x^2$ defined on $C^{\infty}_0(\R)$ is essentially self-adjoint. Its unique self-adjoint extension $-\Delta$ is defined on the Sobolev space $H^2(\R)$, and has only absolutely continuous spectrum $\sigma(-\Delta)=\sigma_{\mathrm{ac}}(-\Delta)=[0,\infty)$, $\sigma_{\mathrm{sc}}(-\Delta)=\sigma_{\mathrm{pp}}(-\Delta)=\emptyset$. 
\begin{lem}
\label{lem:Greenf1} Let $-\Delta$ be the unique self-adjoint extension of $-\partial^2/\partial x^2$. The corresponding resolution of identity $P(E)=\chi_{(-\infty,E)}(-\Delta)$ has kernel 
\be
\bra{x}P(E)\ket{y}=\int_{0}^{\sqrt{E}/\pi} \cos(\pi(x-y)u)\de u.
\ee
In particular,
\be
\bra{x}P(\pi^2)\ket{y}=\frac{\sin(\pi(x-y))}{\pi(x-y)}.
\label{eq:kernelLapl}
\ee
\end{lem}
\begin{proof} Let $G_{z}(x,y)=\bra{x}(-\Delta-z)^{-1}\ket{y}$ be the integral kernel of the resolvent of $-\Delta$. A standard exercise in Fourier coordinates gives, for $\operatorname{Im}z>0$, 
\be
G_{z}(x,y)= i\frac{e^{i|x-y|\sqrt{z}}}{2\sqrt{z}}
\label{eq:G_z}
\ee
so that
\begin{align}
\frac{1}{\pi}\lim_{\eta\downarrow0}\mathrm{Im}\,G_{\epsilon+i\eta}(x,y)&= \frac{\cos((x-y)\sqrt{\epsilon})}{2\pi\sqrt{\epsilon}}1_{\epsilon>0}\quad\text{if $\epsilon\neq0$},
\label{eq:GF_R}\\
\eta\lim_{\eta\downarrow0}\mathrm{Im}\,G_{i\eta}(x,y)&= 0,
\label{eq:GF_R2}
\end{align}
locally uniform in $x,y\in\R$. Then, the following residue formula holds
\be
\bra{x}P(E)\ket{y}=\int_{0}^{E}\frac{1}{\pi}\lim_{\eta\downarrow0}\mathrm{Im}\,G_{\epsilon+i\eta}(x,y)\de \epsilon.
\label{eq:residue_formula}
\ee
and the claim follows by inserting~\eqref{eq:GF_R} in~\eqref{eq:residue_formula} with the change of variables $u=\sqrt{\epsilon}/\pi$.
\end{proof}

Next, we want to write the rescaling of the kernel~\eqref{eq:unified_sink} in terms of the action of a unitary group  on $L^2(\R)$. The affine change of coordinates is given by 
  \begin{alignat}{2}
V_{x_0,E}\colon L^2(\R)&\longrightarrow&& L^2(\R) \nonumber \\
f(x)&\longmapsto&& \sqrt{\frac{2\pi }{N(E)}}f\left(x_0+\frac{2\pi x}{N(E)}\right). 
\label{eq:unitary_affine}
\end{alignat}
Of course $\left(V^{\dagger}_{x_0,E}f\right)(x)=\sqrt{\frac{N(E)}{2\pi}}f\left(\frac{N(E)}{2\pi}(x-x_0)\right)$, and $V_{x_0,E}$ is unitary.

Consider the integral kernel of the spectral projection $\chi_{(-\infty,E)}(H_U)$. In formulae
\be
\bra{x}\chi_{(-\infty,E)}(H_U)\ket{y}=\sum_{E_k\leq E}\overline{\psi_{E_k}}(x)\psi_{E_k}(y),
\ee
where $H_U\psi_{E_k}=E_k\psi_{E_k}$. Let us denote $N(E)=\#\{E_k\leq E\}$. 
If we conjugate the Hamiltonian $H_U$ by the scaling unitary $V_{x_0,E}$, we get that the kernel of the rescaled projection is the rescaled kernel:
\begin{align}
\bra{x}\chi_{(-\infty,E)}(V_{x_0,E}H_UV^{\dagger}_{x_0,E})\ket{y}=\frac{2\pi}{N(E)}\sum_{E_k\leq E}\overline{\psi_{E_k}}\left(x_0+\frac{2\pi x}{N(E)}\right)\psi_{E_k}\left(x_0+\frac{2\pi y}{N(E)}\right),\label{eq:unified_sink2}
\end{align}
so that~\eqref{eq:unified_sink} can be written as (see~\eqref{eq:kernelLapl})
\be
\lim_{E\to\infty}\bra{x}\chi_{(-\infty,E)}(V_{x_0,E}H_UV^{\dagger}_{x_0,E})\ket{y}=\bra{x}\chi_{(-\infty,\pi^2)}(-\Delta)\ket{y},
\ee
for $U\in\{\sigma_1,-I,I,-\sigma_3\}$ (periodic, Dirichlet, Neumann, and Zaremba b.c., respectively).

The next Theorem~\ref{thm:univ_sk} shows that, for any self-adjoint extension of the Laplacian on a finite interval, the family of rescaled projections $\chi_{(-\infty,E)}(V_{x_0,E}H_UV^{\dagger}_{x_0,E})$  converges,  in the strong sense, to the projection $\chi_{(-\infty,\pi^2)}(-\Delta)$ of the (unique) self-adjoint Laplacian on the real line.

\begin{theorem}[The sine kernel for all self-adjoint extensions of the Laplacian]
\label{thm:univ_sk}
For all $U\in\mathrm{U}(2)$ and $x_0\in(0,2\pi)$, 
the following limit holds
\be
\lim_{E\to\infty}
\chi_{(-\infty,E)}(V_{x_0,E}H_UV^{\dagger}_{x_0,E})=\chi_{(-\infty,\pi^2)}(-\Delta),
\label{eq:conj1}
\ee
in the strong sense. 
\end{theorem}
\begin{rmk} Given that  $\bra{x}\chi_{(-\infty,\pi^2)}(-\Delta)\ket{y}$ is the sine kernel~\eqref{eq:kernelLapl}, we expect that the free fermions process converges to the sine process. However, the strong convergence of $\chi_{(-\infty,E)}(V_{x_0,E}H_UV^{\dagger}_{x_0,E})$ does not imply the locally uniform convergence of the kernels $\bra{x}\chi_{(-\infty,E)}(V_{x_0,E}H_UV^{\dagger}_{x_0,E})\ket{y}$. To show the latter convergence, one usually needs to work with quite `explicit' formulae for the eigenfunctions of $H_U$, which are not available for generic quantum boundary conditions.
\qed
\end{rmk}

The idea of the proof is that at microscopic scales in the bulk, the spectral projections of $H_U$ can be approximated arbitrarily well by the spectral projections of the Laplacian $-\Delta$ on $\R$ (the boundary conditions become immaterial). See Fig.~\ref{fig:rescaling_bulk}. The precise way to give a meaning to this approximation is the notion of generalized strong resolvent convergence. This idea has been applied recently by Bornemann~\cite{Bornemann16} to study the possible nontrivial scaling limits of determinantal processes whose kernels are given by spectral projections of self-adjoint Sturm-Liouville operators.
\begin{figure}[t]
\centering
\begin{tikzpicture}
[scale=1.1
]
\draw[red,thick]
  ([shift={(0:0cm)}]0,0) arc (-90:270-5:.6cm);
\draw[thick] (-.02,0) circle (.05cm);
\draw [gray] (-0.13,1.12) -- (-0.13,1.28);
\draw [gray](0.08,1.12) -- (0.08,1.28);
\draw [gray] (-0.13,1.12) -- (0.08,1.12);
\draw [gray] (-0.13,1.28) -- (0.08,1.28);
\draw [ black,dotted] (-0.13,1.12) -- (-2,2.2);
\draw [ black,dotted] (0.08,1.12) -- (2,2.2);
\draw [ black,dotted] (0.08,1.28) -- (.8,2.2);
\draw [ black,dotted] (-0.13,1.28) -- (-.8,2.2);

\draw[red,thick] (-2.5,2.5) -- (-0,2.5);
\draw[red,thick] (0,2.5) -- (2.5,2.5);
\draw[red,thick,dashed] (-3.3,2.5) -- (-2.3,2.5);
\draw[red,thick,dashed] (3.3,2.5) -- (2.3,2.5);

\node at (0,-.5) {{ $S^1\setminus\{0\}$}};
\node at (3.25,2.2) {{ $\R$}};
\end{tikzpicture} 
\caption{After rescaling at the `bulk', the Laplacian on the punctured circle  transforms as in figure in the Laplacian on the real line.  }
\label{fig:rescaling_bulk}
\end{figure}
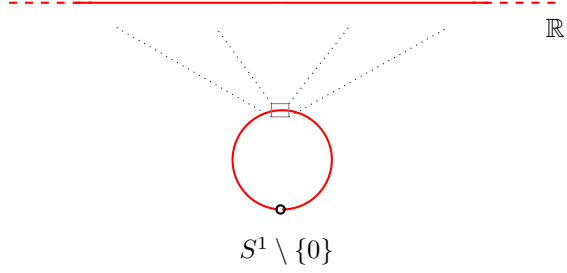
\begin{lem}
\label{lem:Weidmann} Let $(-\Delta_n)_n$ be a sequence of self-adjoint extensions of the formal operator $-\partial^2/\partial x^2$ on $L^2(a_n,b_n)$, and let $-\Delta$ be the unique self-adjoint extension of $-\partial^2/\partial x^2$ defined on $C^{\infty}_0(\R)$. The corresponding resolutions of identities are denoted by $P_n(E)=\chi_{(-\infty,E)}(-\Delta_n)\chi_{(a_n,b_n)}$ and $P(E)=\chi_{(-\infty,E)}(-\Delta)$.  Suppose that $a_n\to-\infty$, $b_n\to\infty$. Then, the sequence $(-\Delta_n)_n$ converges to $-\Delta$ in the strong resolvent sense. In particular, $P_n(E)\to P(E)$ strongly. Moreover, $P(E)$ is left and right continuous, i.e. $P(E_n)\to P(E)$ strongly if $E_n\to E$.
\end{lem}
\begin{proof}
Consider the differential operator $-\partial^2/\partial x^2$ on $(a,b)=(-\infty,\infty)$ and its self-adjoint extension $-\Delta$. Note that i) $-\Delta$ is limit point at $a$ and $b$, and ii) the point spectrum of $-\Delta$ is empty. Then, the strong resolvent convergence $-\Delta_n\stackrel{src}{\to}-\Delta$ is a specialisation of a general result due to Weidmann~\cite{Weidmann97} for self-adjont extensions of formal Sturm-Liouville operators. The fact that $-\Delta_n\stackrel{src}{\to}-\Delta$ implies $P_n(E)\to P(E)$ follows from a classical result essentially due to Rellich. Finally, from the fact that $-\Delta$ has only continuous spectrum, it follows that $P(E)$ is continuous.
\end{proof}
\begin{lem}[Generalised Weyl's law~{\cite[Proposition 4.2]{Bolte09}}] 
\label{lem:Weyl} For all self-adjoint extensions $H^{U}$ of $-\partial^2/\partial x^2$ on $(0,2\pi)$, the number of energy levels $E_k$ (counted with their multiplicities) satisfies the following asymptotic law
\be
N(E)=\#\{E_k\leq E\}=2 \sqrt{E}+O(1),
\label{eq:Weyl}
\ee
 as $E\to\infty$.
\end{lem}
\begin{proof}
As in~\cite{Bolte09}, use the comparison lemma for quadratic forms applied to Dirichlet b.c. ($U=-I$) and  Robin b.c. ($U=e^{i\alpha}I$) which fulfill the asymptotic statement (see~\cite{Ivrii16}).
\end{proof}
\begin{proof}[Proof of Theorem~\ref{thm:univ_sk}]  Fix $U\in\mathrm{U}(2)$ and, therefore a self-adjoint extension $H_U$. The unitary operator $V_{x_0,E}$ defined in~\eqref{eq:unitary_affine} maps wavefunctions in $ L^2(0,2\pi)$ into functions in $L^2(a_E,b_E)$, 
with $a_E=-x_0N(E)/2\pi$ and $b_E=(2\pi-x_0)N(E)/2\pi$. Note that, since $x_0\in(0,2\pi)$, $a_E\to-\infty$ and $b_E\to\infty$, as $E\to\infty$. 
 Consider  the operator $H_U^E$ defined as the original kinetic energy operator $H_U$, but on a rescaled interval:
\barr
D(H_U^{E})\hspace{-2mm}&=&\hspace{-2mm}\left\{\psi\in H^2\left(a_E,b_E\right)\colon\hspace{-1mm}
\left(
\begin{array}{c}
\psi(a_E)+i\psi'(a_E)\\
\psi(b_E)-i\psi'(b_E)\\
\end{array}
\right)\hspace{-1mm}
=
U
\left(
\begin{array}{c}
\psi(a_E)-i\psi'(a_E)\\
\psi(b_E)+i\psi'(b_E)\\
\end{array}
\right)
\right\},
\nonumber 
\\
H_U^E\psi(x)\hspace{-2mm}&=&\hspace{-2mm}-\psi''(x), \quad \psi\in D(H_U^E).
\earr

Then, $V_{x_0,E}$ maps normalized eigenfunctions of $H_U$ into normalized eigenfunctions of $H_U^E$:
\be
H_U^{E}(V_{x_0,E}\psi_{E_k})(x)=\left(\frac{2\pi}{N(E)}\right)^2E_k(V_{x_0,E}\psi_{E_k})(x), \quad \text{for $a_E<x<b_E$},
\ee
and we have the equality of the kernels
\be
\bra{x}\chi_{\left(-\infty,\left(\frac{2\pi}{N(E)}\right)^2E\right)}(H_U^{E})\ket{y}=\bra{x}\chi_{\left(-\infty,E\right)}(V_{x_0,E}H_UV^{\dagger}_{x_0,E})\ket{y}
\ee
By Lemma~\ref{lem:Weidmann}, $H_U^{E}$ approximates the free Laplacian, $H_U^{E}\stackrel{src}{\to}-\Delta$, as $E\to+\infty$. By~\eqref{eq:Weyl} we have
\be
\left(\frac{2\pi}{N(E)}\right)^2E=\pi^2+O\left(\frac{1}{\sqrt{E}}\right).
\ee 
and, again by Lemma~\ref{lem:Weidmann}, we conclude that
\be
\chi_{\left(-\infty,\left(\frac{2\pi}{N(E)}\right)^2E\right)}(H_U^{E})\to \chi_{(-\infty, \pi^2)}(-\Delta),
\ee
in the strong operator sense.
\end{proof}

\subsection{Scaling limits at the edges}
\label{sec:edges} 
\begin{figure}[t]
\centering
\begin{tikzpicture}
[scale=1.1
]
\draw[red,thick]
  ([shift={(0:0cm)}]0,0) arc (-90:270-5:.6cm);
\draw[thick] (-.02,0) circle (.05cm);
\draw [gray] (-0.15,0.1) -- (-0.15,-0.1);
\draw [gray](0.1,0.1) -- (0.1,-0.1);
\draw [gray] (-0.15,0.1) -- (0.1,0.1);
\draw [gray] (-0.15,-0.1) -- (0.1,-0.1);
\draw [ black,dotted] (-0.15,-0.1) -- (-1,-1.5);
\draw [ black,dotted] (0.1,-0.1) -- (1,-1.5);
\draw [ black,dotted] (0.1,0.1) -- (1.8,-1.5);
\draw [ black,dotted] (-0.15,0.1) -- (-1.8,-1.5);
\draw[red,thick] (-2.5,-1.8) -- (-0.05,-1.8);
\draw[red,thick] (0.05,-1.8) -- (2.5,-1.8);
\draw[thick] (0,-1.8) circle (.05cm);
\draw[red,thick,dashed] (-3.3,-1.8) -- (-2.3,-1.8);
\draw[red,thick,dashed] (3.3,-1.8) -- (2.3,-1.8);

\node at (0,1.5) {{ $S^1\setminus\{0\}$}};
\node at (0,-2.25) {{ $\R\setminus\{0\}$}};
\end{tikzpicture} 
\caption{After rescaling at the `edge',  the Laplacian on the punctured circle  transforms, as in figure, in the Laplacian on the punctured line.  }
\label{fig:rescaling_edge}
\end{figure}
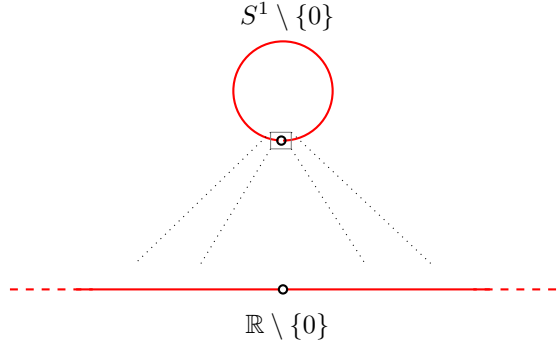
At the edges $0$ and $2\pi$ we do not expect to see a universal scaling limit. The boundary conditions break the translational invariance of the system and introduce a nonuniversal behaviour at the edges. For Dirichlet and Neumann conditions we obtain special cases of the Bessel process~\eqref{eq:D_Bessel}-\eqref{eq:N_Bessel}. By miming the proof of Theorem~\ref{thm:univ_sk} we would like to identify a limiting self-adjoint operator $A$ to which the rescaled Laplacian $H_U^E$ converges in the strong resolvent sense, $H_U^E\stackrel{src}{\to}A$. 

Let $-\Delta_n^{U}$ be a self-adjoint extension of the differential operator $-\partial^2/\partial x^2$ acting on $C^{\infty}_0(r_nS^1\setminus\{e^{i0}\})$, that is a punctured circle of radius $r_n>0$. and $-\Delta^{U}$ a self-adjoint extension of the differential operator $-\partial^2/\partial x^2$ acting on $C^{\infty}_0(\R\setminus\{0\})$. In both case, $U\in\mathrm{U}(2)$ fixes the boundary conditions at $0^+$ and $0^-$. Suppose that $r_n\to\infty$. The set $C^{\infty}_0(\R\setminus\{0\})$ is a core for $-\Delta^{U}$, and every function in $C^{\infty}_0(\R\setminus\{0\})$ is contained, in an obvious way, in the domain of $-\Delta_n^{U}$ for $n$ sufficiently large. By Weidmann's theorem~\cite{Weidmann97}, we have the strong resolvent convergence $-\Delta_n^{U}\stackrel{src}{\to}-\Delta^{U}$. See Fig.~\ref{fig:rescaling_edge}.

We first focus on the case of \emph{local} boundary conditions which do not mix values of the wavefunction and its derivatives at $0^+$ and $0^-$. It is clear that, for local b.c., in the scaling limit at the edge, $-\Delta_n^{U}$ converges to `two' self-adjoint extensions of the Laplacian acting separately on two half-lines $\R_-$ and $\R_+$. Without losing generality, the subset of self-adjoint extensions we are looking for is described by diagonal unitaries of the form $U=e^{i\alpha}I\in\mathrm{U}(2)$; these correspond to Robin b.c., $\psi_\pm'=\pm\tan(\alpha/2)\psi_\pm$, and include Dirichlet and Neumann b.c. as degenerate cases when $\alpha=\pi$ and $\alpha=0$, respectively.

\begin{theorem}[Scaling limit at the edges for local b.c.]
\label{thm:edge_Robin} Let $U=e^{i\alpha}I$ with $\alpha\in(0,\pi)$, and $x_0,y_0\in\{0,2\pi\}$. Set $c=\tan(\alpha/2)$.Then 
\be
\lim_{E\to\infty}
\chi_{(-\infty,E)}(V_{x_0,E}H_UV^{\dagger}_{y_0,E})=\chi_{(-\infty,\pi^2)}(-\Delta^{(c)})1_{x_0=y_0},
\label{eq:limit_edge_R}
\ee
where the integral kernel of $\chi_{(-\infty,\pi^2)}(-\Delta^{(c)})$ is given explicitly by
\begin{align}
\bra{x}\chi_{(-\infty,\pi^2)}(-\Delta^{(c)})\ket{y}&=\hspace{-1mm}\frac{\sin(\pi(x-y))}{\pi(x-y)}+\frac{\sin(\pi(x+y))}{\pi(x+y)}\nonumber\\
&-2c\int_0^{\infty}\frac{\sin(\pi(x+y+\xi))}{\pi(x+y+\xi)}e^{-c\xi}\de u.
\label{eq:scalingRobin}
\end{align}
\end{theorem}
\begin{figure}[t]
\centering
\includegraphics[width=1\textwidth]{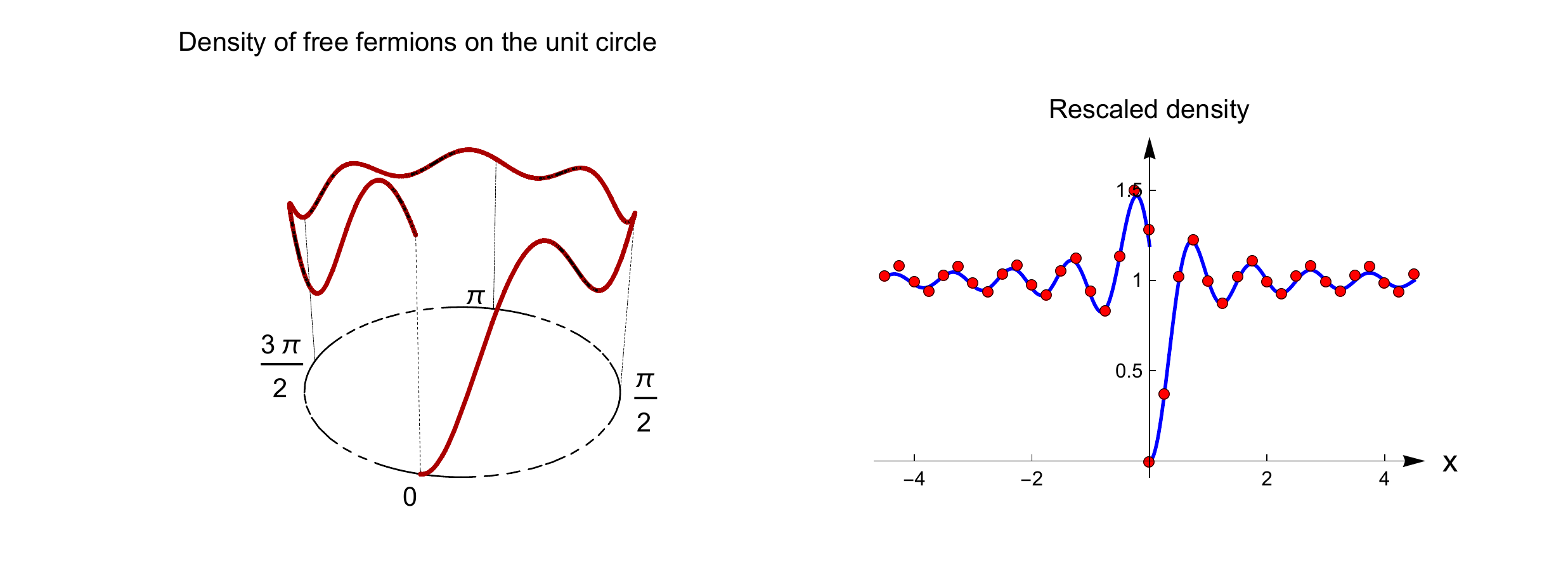}
\caption{
Mixed Dirichlet-Robin boundary conditions: $\psi(0^+)=0$ and $\psi'(0^-)=-\tan(\alpha/2)\psi(0^-)$. Here $\alpha= \pi/2$. Left: particle density on the circle for the ground state of $N=7$ fermions. Right: Rescaled density (red dots) on the left and the right of $0$. Note the different scaling limits (blue solid lines) at $0^{\pm}$ given by \eqref{eq:scalingRobin}.
}
\label{fig:D-R}
\end{figure}
\begin{rmk} The most general case of \emph{local} boundary conditions is given by matrices $U=\diag(e^{i\alpha},e^{i\beta})$. They correspond to different Robin boundary conditions at the edges $0$ and $2\pi$. It is clear that in the scaling limit, the edges are not coupled and, therefore, Theorem~\ref{thm:edge_Robin} covers general local boundary conditions.

Consider, for instance, a free particle in the box with mixed Dirichlet-Robin b.c., i.e. $\psi(0)=0$ and $\psi'(2\pi)=-c\psi(2\pi)$ with $c=\tan(\alpha/2)$. This choice corresponds to take $U=\diag(1,e^{i\alpha})$. The eigenvalues $E_{k}$ and eigenfunctions $\psi_{E_k}$ of  $H_U$ are 
\be
E_k=\omega_k^2, \quad \psi_{E_k}(x)=\sqrt{\frac{4 \omega_k }{4 \pi  \omega_k -\sin (4 \pi  \omega_k )}} \sin (\omega_k x),
\ee
where $\omega_k$ are the nonnegative solutions of the equation $\omega=c\tan\omega$.  See Fig.~\ref{fig:D-R}. Theorem~\ref{thm:edge_Robin} indicates that, if we consider the ground state of $N$ fermions, then at $x_0=0$ the particle density converges to $1-\sin(2\pi x)/(2\pi x)$; at $x_0=2\pi$ the density converges to $1+\sin(2\pi x)/(2\pi x)-2c\int_0^{\infty}\sin(2\pi x+\pi \xi)/(2\pi x+\pi \xi) e^{-c\xi}\de\xi$. This is shown numerically in  Fig.~\ref{fig:D-R}. In the circle geometry, we see convergence to different limits on the right and left of $0$. This is expected as Robin boundary conditions are local.
\qed
\end{rmk}

For \emph{non-local} b.c. the situation is more complicated. In this case, the Laplacian on $\R\setminus\{0\}$ does not `decouple' into the the two half-lines, and one needs to consider genuine singular perturbations of Schr\"odinger type operators. We focus on the boundary conditions, usually denoted in physics as $\delta$-perturbations of the Laplacian, $\psi_+=\psi_-$ and $\psi_+'-\psi_-'=c\psi_+$. These include the case of periodic b.c. ($c=0$).
\begin{theorem}[Scaling limit at the edges for delta potentials] 
\label{thm:edge_delta}
Let $U=(1-ic/2)^{-1}(\sigma_1-ic I/2)$ (free particle with $\delta$-perturbation), and $x_0,y_0\in\{0,2\pi\}$. Then 
\be
\chi_{(-\infty,E)}(V_{x_0,E}H_UV^{\dagger}_{y_0,E})=\chi_{(-\infty,\pi^2)}(-\Delta+c\delta),
\ee
where the integral kernel is
\be
\bra{x}\chi_{(-\infty,\pi^2)}(-\Delta+c\delta)\ket{y}=\frac{\sin(\pi(x-y))}{\pi(x-y)}+c\int_0^{1}\frac{\sin(\pi(x+y)u)}{2\pi u+c}\de u.
\label{eq:kernel_delta}
\ee
\end{theorem}
\begin{rmk} Note that for Robin b.c., the limit integral operator~\eqref{eq:limit_edge_R} at the edges is non trivial if and only if $x_0=y_0$. This is expected, as local b.c. do not couple the edges $0$ and $2\pi$. The situation is different in the case of $-\Delta+c\delta$ (and other non local b.c.), where a nontrivial limit exists even in the case $x_0=0$ and $y_0=2\pi$. Note that for $\alpha=\pi$ and $\alpha=0$ in~\eqref{eq:limit_edge_R} one obtains the Bessel kernels of Dirichlet and Neumann b.c., respectively. Kernels similar to~\eqref{eq:limit_edge_R}, have been considered by Johansson as variants of Dyson's Hermitian Brownian motion after a finite time~\cite{Johansson04}. Setting $c=0$ in ~\eqref{eq:kernel_delta}, we are back to the case of periodic b.c. (sine kernel).
\qed
\end{rmk}
The scheme of the proof of the above results is similar to the previous (cf. the proof of  Theorem~\ref{thm:univ_sk}), so we omit some details. For Theorem~\ref{thm:edge_Robin}, what we need is the integral kernel of the resolvent of the self-adjoint extensions of $-\partial^2/\partial x^2$ acting on $C^{\infty}_0(\R_+)$. The self-adjoint extension of this symmetric operator are parametrized by unitary matrices from $\mathrm{U}(m)$, where $m$ is the deficiency index. Is is known that $m=1$, and therefore (not surprisingly) the self-adjoint extensions of  $-\partial^2/\partial x^2$ acting on  the half-line are labelled by one real parameter ($\dim_R\mathrm{U}(1)=1$) that specifies the behaviour of the wavefunctions at the boundary point $0$.
For Theorem~\ref{thm:edge_delta}, we are led to consider the resolvent of the self-adjoint extension `` $-\Delta+c\delta$ '' of $-\partial^2/\partial x^2$ acting on $C^{\infty}_0(\R\setminus\{0\})$ (the punctured line). 

Theorems~\ref{thm:edge_Robin} and~\ref{thm:edge_delta} follow from the following lemmas.
\begin{lem}
Consider the Laplacian operator on the half-line with Robin boundary conditions
\barr
D(-\Delta^c)&=&\left\{\psi\in H^2\left(\R_+\right)\colon 
\psi'(0^+)=c\psi(0^+)
\right\}
\\
-\Delta^c\psi(x)&=&-\psi''(x),\quad \psi\in D(-\Delta^c),
\earr
with\footnote{Note that the discrete spectrum of $-\Delta^c$ is empty for $c>0$. See~\cite[Eq. (2.13)]{Albeverio95} for details.} $c>0$. Then, the resolution of identity $P(E)=\chi_{(-\infty,E)}(-\Delta^c)$ has kernel 
\barr
\bra{x}P(E)\ket{y}&=&\int_0^{\sqrt{E}/\pi}\biggl(\cos(\pi(x-y)u)+\cos(\pi(x +y)u)\nonumber\\
&&-2\frac{c^2\cos(\pi(x+y)u)-c\pi u\sin(\pi(x+y)u)}{\pi^2u^2+c^2}\biggr)\de u .
\earr
\end{lem}
\begin{proof} The integral kernel of the resolvent $\bra{x}(-\Delta^c-z)^{-1}\ket{y}$ can be obtained as Laplace transform in the time variable $t$ of the transition probability $p_t^c(x,y)=\bra{x}e^{-\Delta^c t}\ket{y}$. The latter, is nothing but the heat kernel of a Brownian motion\footnote{Different self-adjoint extensions of the Laplacian correspond to generators of different Markov processes. The classical boundary conditions of Dirichlet, Neumann and Robin correspond respectively to a killed, reflected, and partially reflected Brownian motion at the boundary (see~\cite{Papanicolaou88}). } (or quantum propagator at imaginary time) on the half-line with Robin boundary condition. It can be found by the method of images, which amounts to extend the problem on the line (where the heat kernel is known) using a suitable reflection that fixes the boundary conditions at $x=0$. For Robin b.c., one finds
\be
p_t^c(x,y)=p_t(x-y)+p_t(x+y)-2c\int_{0}^{+\infty}p_t(x+y+\xi)e^{-c\xi}\de \xi,
\ee
where $p_t(x,y)=\bra{x}e^{-\Delta t}\ket{y}$ is the transition probability of the process generated by the free Laplacian on $\R$, i.e. the heat kernel of the Brownian motion on the line (or free propagator at imaginary time). 
Therefore, we have
\barr
\bra{x}(-\Delta^c-z)^{-1}\ket{y}= G_{z}(x,y)+G_{z}(x,-y)-2c\int_{0}^{+\infty}G_{z}(x,-y-\xi)e^{-c\xi}\de \xi,
\earr
with $G_{z}(x,-y)$ given in~\eqref{eq:G_z}. Performing the elementary integration on $\xi$ (note that $x$, $y$, and $\xi$ are positive), and using the formula
\be
\bra{x}P(E)\ket{y}=\int_{0}^{E}\frac{1}{\pi}\lim_{\eta\downarrow0}\mathrm{Im}\,\bra{x}(-\Delta^c-z)^{-1}\ket{y}\de \epsilon,
\label{eq:residue_formula2}
\ee
we conclude the proof. 
\end{proof}

\begin{lem}
Consider $-\partial^2/\partial x^2$ acting on $C^{\infty}_0(\R\setminus\{0\})$, and denote by $-\Delta+c\delta$ its self-adjoint extension defined by the boundary conditions $\psi(0^+)=\psi(0^-)$ and $\psi'(0^+)-\psi'(0^-)=c\psi(0^-)$. 
Then, the integral kernel of the spectral projection $P(E)=\chi_{(-\infty,E)}(-\Delta+c\delta)$,
\be
\bra{x}P(E)\ket{y}=\int_0^{\sqrt{E}/\pi}\left(\cos(\pi(x-y)u)+\frac{c}{2\pi u+c}\sin(\pi(|x|+|y|)u)\right)\de u.
\label{eq:spectProj_delta}
\ee
\end{lem}
\begin{proof}
The integral kernel of the resolvent $(-\Delta+c\delta)^{-1}$ can be computed by Krein's formula. The explicit expression is~\cite{Albeverio95}
\be
\bra{x}(-\Delta+c\delta-z)^{-1}\ket{y}=G_z(x,y)-\frac{2\sqrt{z}c}{2\sqrt{z}+c}G_z(x,0)G_z(y,0),
\ee
where $G_z(x,y)$ is the free space resolvent~\eqref{eq:G_z}, so that
\be
\bra{x}(-\Delta+c\delta-z)^{-1}\ket{y}=\frac{1}{2\sqrt{z}c}\left(ie^{i|x-y|\sqrt{z}}+\frac{c}{2\sqrt{z}+c}e^{i(|x|+|y|)\sqrt{z}}\right).
\ee
For $-\Delta+c\delta$, the essential spectrum coincides with the absolutely continuous spectrum and is equal to $[0,\infty)$. The singular spectrum and the discrete spectrum are empty. Therefore, using residues formula we obtain the integral kernel~\eqref{eq:spectProj_delta}. 
\end{proof}
\section{Grand canonical processes at finite temperature}
\label{sec:CUET}
We now  extend to finite temperature the determinantal process analised in the previous Section.  We start from the `easiest' case, namely the CUE process with correlation kernel $Q_{\mathrm{U}(2N+1)}(x,y)$. It is familiar to those working in random matrix theory, that the CUE enjoys some algebraic simplifications compared to the GUE process, and the microscopic (universal) behaviour of the eigenvalues can be obtain in an easier way than for the GUE. Indeed, this was one of the motivations for Dyson to introduce  the CUE in random matrix theory. We shall see that the same simplifications persist at $T>0$.
\subsection{Finite temperature CUE}
We propose a finite temperature CUE defined (in analogy to $\mathrm{GUE}(T,\mu)$) as the grand canonical process with correlation kernel
\be
K_{\mathrm{CUE}(T,\mu)}(x,y)=\frac{1}{2\pi}\sum_{k\in\Z}\frac{e^{ ik(x-y)}}{1+e^{-(\mu-k^2)/T}},
\label{eq:CUET}
\ee
The chemical potential $\mu=\mu(N,T)$  may be chosen from the condition $\int_0^{2\pi}K_{\mathrm{CUE}(T,\mu)}(x,x)\de x=2N+1$, i.e.,
\be
2N+1=\sum_{k\in\Z}\frac{1}{1+e^{-(\mu-k^2)/T}}.
\label{eq:CUEconstr}
\ee
(Note that $K_{\mathrm{CUE}(T,\mu)}(x,y)$ defines a trace class operator.) Linear statistics on finite temperature extensions of the CUE (with generic \emph{shape functions} other than the Fermi factor) have been recently studied by Johansson and Lambert~\cite{Lambert15}. For all $T$, the one-point correlation function is, of course, constant on the interval of length $2\pi$,
\be
K_{\mathrm{CUE}(T,\mu)}(x,x)=\frac{1}{2\pi}\sum_{k\in\Z}\frac{1}{1+e^{-(\mu-k^2)/T}}=\frac{2N+1}{2\pi},
\label{eq:K(x,x)}
\ee
by virtue of~\eqref{eq:CUEconstr}.
(Note, in contrast, that the finite temperature GUE undergoes a transition from the semicircular law to a Gaussian.)

The finite temperature CUE~\eqref{eq:CUET}-\eqref{eq:CUEconstr} interpolates between $N$ independent random variables on the circle and eigenvalues of matrices from the CUE ensemble.
The next theorem is the analogue of~\eqref{eq:Johansson1}-\eqref{eq:Johansson3} of the finite temperature GUE.
\begin{theorem}\label{thm:finiteT_CUE} Let $K_{\mathrm{CUE}(T,\mu)}(x,y)$ be as in~\eqref{eq:CUET}-\eqref{eq:CUEconstr}. Then,
\begin{itemize}
\item[i)] Interpolation between Poisson and CUE: if $\mu(N,T)=N^2$,
\be
\lim_{T\to 0}K_{\mathrm{CUE}(T,\mu)}(x,y)=Q_{\mathrm{U}(2N+1)}(x,y)
\label{eq:Tto0}
\ee
uniformly for $x,y$ in a compact set; if $\mu(N,T)=T\log\left(\frac{2N+1}{\sqrt{\pi T}}\right)$, then
\be
\lim_{T\to\infty}K_{\mathrm{CUE}(T,\mu)}(x,y)=
\left\{  \begin{array}{l@{\quad}cr} 
0&\text{if $x\neq y$,}\\
\frac{2N+1}{2\pi}&\text{if $x=y$,}
\end{array}\right.
\label{eq:Ttoinfty}
\ee
pointwise. 
\item[ii)] Scaling limit of high temperature and large number of particles in the bulk:\\
Let $T=c N^2$ and $\mu=cN^2\log{\lambda}$ with $c>0$, and set $\lambda=\operatorname{Li}_{1/2}^{-1}(-2/\sqrt{\pi c})$. Then, the following limit holds
\be
\lim_{N\to\infty}\frac{\pi}{N}K_{\mathrm{CUE}(cN^2,cN^2\log{\lambda})}\left(\frac{\pi x}{N},\frac{\pi y}{N}\right)=\int_0^{\infty}\frac{\cos{(\pi(x-y)u)}}{1+\lambda^{-1}e^{u^2/c}}\,\de u.
\label{eq:bulkCUE}
\ee
uniformly for $x,y$ in a compact set.
\end{itemize}
The conditions on $\mu(N,T)$ in~\eqref{eq:Tto0}-\eqref{eq:Ttoinfty} provide approximate solutions of the constraint~\eqref{eq:CUEconstr} on the number of particles in the appropriate regimes of temperature. The coiche of the parameter $\lambda$ also provide an approximate solution of~\eqref{eq:CUEconstr} (see formula~\eqref{eq:formula_lambda} below). 
\end{theorem}
Since the system is periodic, there are no edges (no analogue of the finite temperature Airy kernel). Note that the limit kernel in the bulk~\eqref{eq:bulkCUE} is the same as for the finite temperature GUE (universality). 
\begin{proof} 
To prove~\eqref{eq:Tto0} note that
\be
\lim_{T\to0}\frac{1}{1+e^{-(N^2-k^2)/T}}=\chi_{(-\infty,N)}(|k|).
\ee
Computing the limit we find
\barr
\lim_{T\to0}K_{\mathrm{CUE}(T,N^2)}(x,y)
&&=\sum_{|k|\leq N}\frac{e^{ik(x-y)}}{2\pi}\chi_{(-\infty,N)}(|k|)+\sum_{|k|> N}\frac{e^{ik(x-y)}}{2\pi}\chi_{(-\infty,N)}(|k|)\nonumber\\
&&=\sum_{|k|\leq N}\frac{e^{ik(x-y)}}{2\pi}=Q_{\mathrm{U}(2N+1)}(x,y).
\earr

Now we prove~\eqref{eq:Ttoinfty}.  Set $\mu=T\log\left(\frac{2N+1}{\sqrt{\pi T}}\right)$. For $x=y$, by monotonicity, we have the bound
\begin{align}
&\left|K_{\mathrm{CUE}(T,\mu)}(x,x)-\frac{\sqrt{T}}{2\pi}\int_{-\infty}^{+\infty}\frac{1}{1+\frac{\sqrt{\pi T}}{2N+1}e^{u^2}}\de u
\right|\nonumber\\
&\leq\frac{1}{2\pi}\frac{1}{1+\frac{\sqrt{\pi T}}{2N+1}}+\frac{1}{\pi}\left|\sum_{k=0}^{\infty}\frac{1}{1+\frac{\sqrt{\pi T}}{2N+1}e^{k^2/T}}
-\sqrt{T}\int_{0}^{+\infty}\frac{1}{1+\frac{\sqrt{\pi T}}{2N+1}e^{u^2}}\de u\right|\nonumber\\
&\leq\frac{3}{2\pi}\frac{1}{1+\frac{\sqrt{\pi T}}{2N+1}}\stackrel{T\to\infty}{\longrightarrow0}.
\end{align}
By dominated convergence,
\be
\lim_{T\to\infty}\frac{\sqrt{T}}{2\pi}\int_{-\infty}^{+\infty}\frac{1}{1+\frac{\sqrt{\pi T}}{2N+1}e^{u^2}}\de u=\frac{2N+1}{2\pi}\int_{-\infty}^{+\infty}\frac{e^{-u^2}}{\sqrt{\pi}}\de u=\frac{2N+1}{2\pi},
\ee
For $x\neq y$, 
\begin{align}
&\left|K_{\mathrm{CUE}(T,\mu)}(x,y)-\frac{\sqrt{T}}{2\pi}\int_{-\infty}^{+\infty}\frac{e^{iu\sqrt{T}(x-y)}}{1+\frac{\sqrt{\pi T}}{2N+1}e^{u^2}}\de u
\right|\nonumber\\
&\leq \frac{1}{2\pi}\frac{1}{1+\frac{\sqrt{\pi T}}{2N+1}}+\frac{1}{\pi}\left|\sum_{k=0}^{\infty}\frac{\cos(k(x-y))}{1+\frac{\sqrt{\pi T}}{2N+1}e^{k^2/T}}
-\sqrt{T}\int_{0}^{+\infty}\frac{\cos(\sqrt{T}(x-y)u)}{1+\frac{\sqrt{\pi T}}{2N+1}e^{u^2}}\de u\right|\nonumber\\
&\leq \frac{1}{2\pi}\frac{1}{1+\frac{\sqrt{\pi T}}{2N+1}}+\frac{2}{\pi}\sum_{n=0}\frac{1}{1+\frac{\sqrt{\pi T}}{2N+1}\exp\left(\frac{\pi^2n^2}{T(x-y)^2}\right)}.
\end{align}
In the last inequality we use the fact that the oscillating function $\frac{1}{1+\frac{\sqrt{\pi T}}{2N+1}e^{u^2}}\cos(\sqrt{T}(x-y)u)$ is monotonic in the intervals $u\in \left[n \frac{\pi}{\sqrt{T}(x-y)},(n+1) \frac{\pi}{\sqrt{T}(x-y)}\right]$, $n\in\Z$. The convergent series can be bounded as
\be
\sum_{n=0}^{\infty}\frac{1}{1+\frac{\sqrt{\pi T}}{2N+1}\exp\left(\frac{\pi^2n^2}{T(x-y)^2}\right)}\leq\frac{2N+1}{\sqrt{\pi T}}\sum_{n=0}^{\infty}e^{-\frac{\pi^2}{T(x-y)^2}n}=\frac{2N+1}{\sqrt{\pi T}}\frac{1}{1+e^{-\frac{\pi^2}{T(x-y)^2}}},
\ee
and hence goes to zero as $T\to\infty$.
We write
\barr
\frac{\sqrt{T}}{2\pi}\int_{-\infty}^{+\infty}\frac{e^{iu\sqrt{T}(x-y)}}{1+\frac{\sqrt{\pi T}}{2N+1}e^{u^2}}\de u
=I_1+I_2,
\earr
with
\begin{align}
I_1&=\frac{2N+1}{2\pi^{3/2}}\int_{-\infty}^{+\infty}e^{iu\sqrt{T}(x-y)}e^{-u^2}\de u\\
I_2&=\frac{\sqrt{T}}{2\pi}\int_{-\infty}^{+\infty}\left(\frac{1}{1+\frac{\sqrt{\pi T}}{2N+1}e^{u^2}}-\frac{2N+1}{\sqrt{\pi T}}e^{-u^2}\right)e^{iu\sqrt{T}(x-y)}\de u.
\end{align}
The first integral can be computed
\be
I_1=\frac{2N+1}{2\pi}e^{-\frac{T}{4}(x-y)^2}\stackrel{T\to\infty}{\longrightarrow0}.
\ee 
The second integral is bounded in absolute value
\be
|I_2|\leq\frac{2N+1}{2\pi^{3/2}}\int_{-\infty}^{+\infty}\left(1-\frac{1}{1+\frac{2N+1}{\sqrt{\pi T}}e^{-u^2}}\right)e^{-u^2}\de u,
\ee 
and goes to zero as $T\to\infty$ by monotone convergence.

We proceed now to the proof of~\eqref{eq:bulkCUE}. With the scaling $T=c N^2$, $\mu=cN^2\log{\lambda}$, the constraint on the particle numbers reads 
\barr
\frac{1}{2N}\sum_{k\in\Z}\frac{1}{1+\lambda^{-1} \,e^{\frac{(k/N)^2}{c}}}-\frac{1}{2N}
\stackrel{N\to\infty}{\longrightarrow}\frac{1}{2}\int_{-\infty}^{+\infty}
\frac{\de u}{1+\lambda^{-1} \, e^{\frac{u^2}{c}}}=1.
\label{eq:formula_lambda}
\earr
This explains the condition $\mathrm{Li}_{1/2}(-\lambda)=-2/\sqrt{\pi c}$. Using elementary steps, we find
\be
\frac{\pi}{N}K_{\mathrm{CUE}(cN^2,cN^2\log{\lambda})}\left(\frac{\pi x}{N},\frac{\pi y}{N}\right)=\frac{1}{2 N}\sum_{k\in\Z}\frac{e^{ i\pi\frac{k}{N}(x-y)}}{1+\lambda^{-1}\, e^{\frac{(k/N)^2}{c}}}
\stackrel{N\to\infty}{\longrightarrow}\int_{0}^{\infty}\frac{\cos{(\pi(x-y)u)}}{1+\lambda^{-1}e^{u^2/c}}\,\de u.\nonumber
\ee
\end{proof}
\begin{figure}[t]
\centering
\includegraphics[width=1\textwidth]{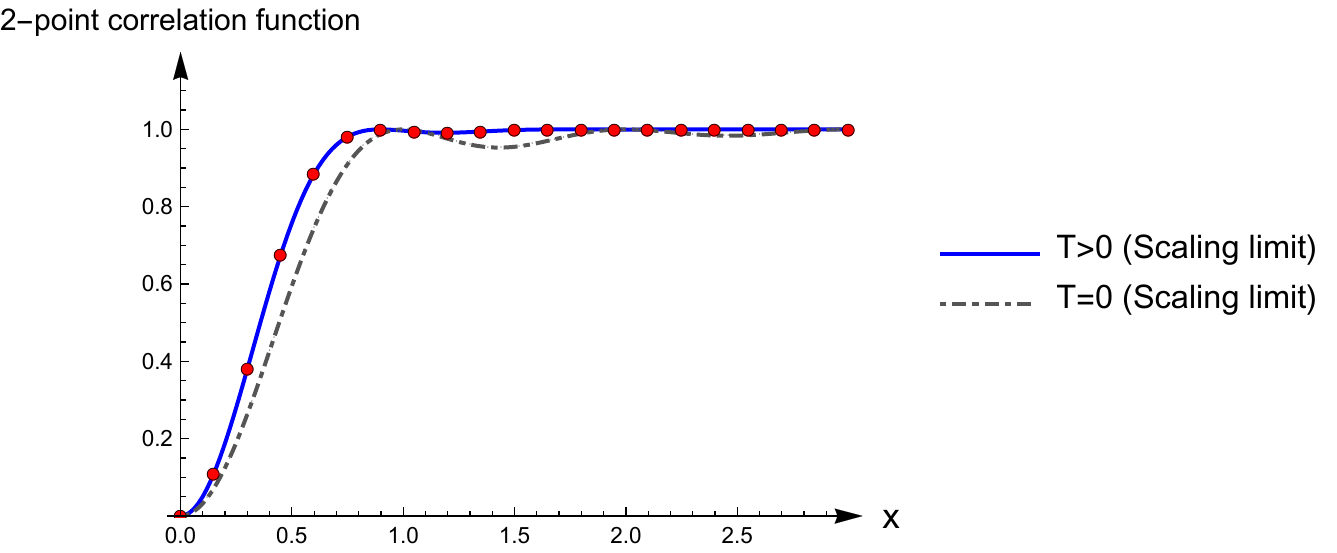}
\caption{
Two-point correlation function for free fermions in a box with periodic boundary conditions. Solid line:  the finite temperature extension of sine-kernel (see Eq.~\eqref{eq:bulkCUE}); Dotted-dashed line: two-point correlation function of the sine process. In red dots: rescaled two-point correlation at finite temperature computed numerically for $N=10$ fermions with $\lambda=10$. See Theorem~\ref{thm:finiteT_CUE} for details.
}
\label{fig:finite_T}
\end{figure}
\subsection{Finite temperature processes for generic self-adjoint extensions}
\label{sec:finiteTsaext}
There is an obvious way to extend the above construction to the other classical groups. Consider a system of free fermions in a box with Dirichlet, Neumann and Zaremba b.c., and construct the determinantal processes defined by the grand canonical correlation kernels
\begin{align}
K^{D}_{T,\mu}(x,y)&=\displaystyle\frac{1}{2\pi}\sum_{k\in\Z}\frac{\sin(kx/2)\sin(ky/2)}{1+e^{-(\mu-k^2/4)/T}},
\label{eq:KDT}\\
K^{N}_{T,\mu}(x,y)&=\displaystyle\frac{1}{2\pi}\sum_{k\in\Z}\frac{\cos(kx/2)\cos(ky/2)}{1+e^{-(\mu-k^2/4)/T}},
\label{eq:KNT}\\
K^{Z}_{T,\mu}(x,y)&=\displaystyle\frac{1}{2\pi}\sum_{k\in\Z+\frac{1}{2}}\frac{\sin(kx/2)\sin(ky/2)}{1+e^{-(\mu-k^2/4)/T}},
\label{eq:KDNT}
\end{align}
where $\mu=\mu(N,T)$ is fixed by the condition 
\be
N=\sum_{k}\frac{1}{1+e^{-(\mu-k^2/4)/T}}.
\label{eq:Groupconstr}
\ee
These kernels provide the natural extension to finite temperature of the eigenvalue process of the classical compact groups $\mathrm{Sp}(2N)$, $\mathrm{SO}(2N)$, and $\mathrm{SO}(2N+1)$. If we denote the Fermi factor by
\be
F_{T,\mu}(z)=\frac{1}{1+e^{-(\mu-z)/T}},
\label{eq:Fermi_factor}
\ee
the correlation kernels~\eqref{eq:KDT}-\eqref{eq:KNT}-\eqref{eq:KDNT} are the integral kernels of the self-adjoint operators $F_{T,\mu}(H_U)$, with $U=-I,I,-\sigma_3$, respectively. 

It is natural to consider the finite temperature kernel associated to $H_U$ ($U\in\mathrm{U}(2)$), for generic boundary conditions
\be
\bra{x}F_{T,\mu}(H_U)\ket{y}=\sum_{E_k}\frac{\overline{\psi_{E_k}}(x)\psi_{E_k}(y)}{1+e^{-(\mu-E_k)/T}},
\ee
where $\mu=\mu(N(E),T)$ is fixed by the condition
\be
\Tr\, F_{T,\mu}(H_U)=N(E).
\label{eq:condition_mu_sa}
\ee
($N(E)=\Tr\,\chi_{(-\infty,E)}(H_U)$ denotes the integrated density of states of $H_U$.)

Irrespectively of the boundary conditions, the grand canonical process of non-interacting free fermions is a kind of interpolation between Poisson ($T\to\infty$) and random matrix statistics ($T\to0$). In the bulk, we expect that the rescaled processes converge to the finite temperature sine process
\be
\lim_{N\to\infty}\frac{2\pi}{N(E)}\sum_{E_k}\frac{\overline{\psi_{E_k}}\left(x_0+\frac{2\pi x}{N(E)}\right)\psi_{E_k}\left(x_0+\frac{2\pi y}{N(E)}\right)}{1+\lambda^{-1}e^{E_k/E}}=\int_0^{\infty}\frac{\cos{(\pi(x-y)u)}}{1+\lambda^{-1}e^{u^2/c}}\,\de u.
\ee
\begin{theorem}[Finite temperature free fermions with generic boundary conditions]
\label{thm:univ_skT}
Let  $U\in\mathrm{U}(2)$. Then,
\item[i)] Interpolation between Poisson and Fermionic process:
\barr
\lim_{T\to 0}\bra{x}F_{T,E}(H_U)\ket{y}&=&\bra{x}\chi_{(-\infty,E)}(H_U)\ket{y},
\label{eq:Tto0_sa}\\
\lim_{T\to\infty}\bra{x}F_{T,T\log\left(\frac{N(E)}{\sqrt{\pi T}}\right)}(H_U)\ket{y}&=&\
\left\{  \begin{array}{l@{\quad}cr} 
0&\text{if $x\neq y$,}\\
\frac{N(E)}{2\pi}&\text{if $x=y$,}
\end{array}\right.
\label{eq:Ttoinfty_sa}
\earr
\item[ii)] Universal scaling limit in the bulk:\\
Let $x_0\in(0,2\pi)$, $T=cE$ and $\mu=cE\log{\lambda}$ with $c>0$. Set $\lambda=-\mathrm{Li}^{-1}_{1/2}(-2/(\sqrt{\pi c})$. Let $V_{x_0,E}$ be the unitary operator defined in~\eqref{eq:unitary_affine}. Then, the following scaling limit holds
\be
\lim_{E\to\infty}
F_{cE,cE\log{\lambda}}(V_{x_0,E}H_UV^{\dagger}_{x_0,E})=F_{c,c\log{\lambda}}(-\Delta),
\label{eq:conjfiniteT}
\ee
in the strong sense. The operator $F_{c,c\log{\lambda}}(-\Delta)$ has kernel
\be
\bra{x}F_{c,c\log{\lambda}}(-\Delta)\ket{y}=\int_0^{\infty}\frac{\cos{(\pi(x-y)u)}}{1+\lambda^{-1}e^{u^2/c}}\,\de u.
\label{eq:bulk_sa}
\ee
\begin{proof} First, notice that the Fermi factor $F_{T,\mu}(z)$ defined in~\eqref{eq:Fermi_factor} is a continuous function satisfying $F_{T,\mu}(z)\leq e^{\mu/T}e^{-z/T}$. Therefore, the argument to prove~\eqref{eq:Tto0_sa}-\eqref{eq:Ttoinfty_sa}  is identical to the one used in the proof of Theorem~\ref{thm:finiteT_CUE}.  

For the second part of the theorem, recall that $N(E)\sim 2\sqrt{E}$, so that the Fermi energy $E(N)$ (the generalised inverse of the density of states $N(E)$) has the asymptotic behaviour $E(N)\sim N^2/4$. 
The condition on the trace
\be
\sum_{E_k}\frac{1}{1+\lambda^{-1}e^{E_k/(cE)}}=N(E)
\ee
explains the choice of $\lambda=-\mathrm{Li}^{-1}_{1/2}(-2/(\sqrt{\pi c})$, that is the solution of 
\be
\int_{0}^{+\infty}\frac{\de u}{1+\lambda^{-1}e^{u^2/c}}=1.
\ee
The proof of~\eqref{eq:conjfiniteT} follows almost verbatim the proof of Theorem~\ref{thm:univ_sk}.  The strong resolvent convergence and the fact that the Fermi factor is a continuous bounded function imply the  convergence as $E\to\infty$ of $F_{cE,cE\log{\lambda}}(V_{x_0,E}H_UV^{\dagger}_{x_0,E})$ to $F_{c,c\log{\lambda}}(-\Delta)$ in the strong sense.
\end{proof}
\end{theorem}
\section{Canonical measures, matrix models and non-intersecting paths}
\label{sec:MM}
In this Section we aim to obtain matrix models whose eigenvalue statistics correspond to the finite temperature processes with kernels 
$F_{T,\mu}(H_U)$, when $U$ corresponds to periodic, Dirichlet, Neumann, and Zaremba boundary conditions (see Problem~\ref{prob:3}). We can legitimately dub those matrix models as `finite temperature extensions' of the Haar measures on the classical compact groups. To define these new matrix ensembles we proceed by analogy to the MNS model (finite temperature extension of the GUE ensemble).

\subsection{The MNS model revisited}
The key observation is that it is possible to write the MSN measure~\eqref{eq:MSNmodel} in the more insightful fashion
\be
P_{n,t}(X)\de X=C_{n,t}e^{-\frac{1}{2}\Tr X^2}\left(\int_{\mathrm{U}(n)}k_t(X-VXV^{\dagger})\de V\right)\de X,
\label{eq:MNS_revisited}
\ee
where $k_t(X)=\exp(-\frac{1}{2t}\Tr X^2)$ is the heat kernel on the algebra of Hermitian matrices, i.e. the fundamental solution of the heat equation. Equation~\eqref{eq:MNS_revisited} corresponds to the evolution in time of a GUE random matrix along the heat flow. The final point at time $t$ is, with `equal probability', any matrix $VXV^{\dagger}$ with the same spectrum of $X$. The diagonalisation of $X$, induces the probability measure~\eqref{eq:MSN_jpdf} on the eigenvalues. It can be write as
\barr
p_{n,t}(x_1,\dots,x_n)=\frac{1}{Z_{n,t}}\det\left[p_t(x_i,x_j)\right]_{i,j=1}^n \prod_{i=1}^ne^{-\frac{1}{2}x_i^2},
\label{eq:MSN_jpdf2}
\earr
where $p_t(x,y)=\bra{x}e^{t\Delta}\ket{y}=e^{-\frac{1}{2t}(x-y)^2}$ is the heat kernel (free propagator at imaginary time) on the real line. We can attach a probabilistic interpretation of~\eqref{eq:MSN_jpdf2} in terms of non-intersecting paths~\cite{Johansson07}. Consider $n$ standard Brownian motions on the real line started at $x_1,\dots,x_n$ at time $0$ (for the Brownian motion, the transition probability is $p_t(x,y)$), conditioned to come back to  $x_1,\dots,x_n$ at time $t$ and without having had any collisions during this time interval.  By a general theorem of Karlin and McGregor~\cite{Karlin59}, the corresponding transition probability is proportional to $\det\left[p_t(x_i,x_j)\right]_{i,j=1}^n$. Put an initial density $\prod_{i=1}^ne^{-x_i^2/2}$ on the initial points $x_1,\dots,x_N$.  We can think of~\eqref{eq:MSN_jpdf2} as a model of non-intersecting paths on a cylinder. There is also an interpretation in terms of non-intersecting Ornstein-Uhlenbeck processes~\cite{LeDoussal17}. 

\subsection{Group heat kernel and non-intersecting loops}
\label{sec:HKU}
It is tempting to generalise the MNS construction to the classical compact groups. Starting from unitary matrices, we consider the unitarily invariant ensemble of matrices in $\mathrm{U}(2n+1)$ defined by the measure
\be
P_{n,t}(U)\de U=\frac{1}{\mathcal{Z}_{n,t}}\left(\int_{\mathrm{U}(2n+1)}K_t(U(VUV^{-1})^{-1})\de V\right)\de U.\label{eq:MSNunit}
\ee
In the above formula, $K_t(g)$ denotes the \emph{group} heat kernel (defined below).

In analogy to the MNS model, we consider the quantum propagator $\bra{x}e^{-tH_{\sigma_1}}\ket{y}$ of a free particle in a box with periodic boundary conditions 
\be
p_t(x,y)=\sum_{k\in\Z}e^{-E_k t}\overline{\psi_k(x)}\psi_k(y)=\frac{1}{2\pi}\sum_{k\in\Z}e^{-k^2 t+ik(x-y)}=\frac{1}{2\pi}\Theta\left(\frac{x-y}{2\pi}, \frac{i t}{\pi}\right),
\label{eq:heat_kernel_CUE}
\ee
where the Jacobi theta function $\Theta(z,\tau)$ is defined by the series
\be
\Theta(z,\tau)=\sum_{k\in\Z}e^{\pi i \tau k^2}e^{2\pi i k z},
\label{eq:Jacobi}
\ee
which converges for all $z\in\C$ and $\mathrm{Im}\tau>0$. Note that $p_t(x,y)$ is the transition density function of a Brownian motion on a circle, i.e., the probability that a Brownian particle moves from $x$ to $y$ in a time $t$. This formula may be derived as the fundamental solution of the heat equation on the circle. By a theorem of  Karlin and McGregor~\cite[Theorem 1 and Ex. (iv)]{Karlin59}, all the odd determinants of $p_t(x,y)$ are strictly positive. In particular, if $0\leq x_1<x_2<\cdots<x_{2N+1}\leq 2\pi$, then $\det\left(p_t(x_i,x_j)\right)_{i,j=1}^{2n+1}\geq0$.
Consider $2n+1$ standard Brownian motions on the circle started at $x_1,\dots,x_{2n+1}$ at time $0$, conditioned to come back to $x_1,\dots,x_{2n+1}$ at time $t$ and without having had any collisions during this time. Put an initial uniform density $\prod_{i=1}^{2n+1}\frac{\de x_i}{2\pi}$ on the points $x_1,\dots,x_{2n+1}$. Then, we get a probability measure\footnote{This is a special case of the Karlin and McGregor formula when the state space is a circle and the number of particles is odd, since the cyclic permutations of an odd number of objects are all even permutations. When the number of particles is even, a similar probability measure can be constructed but it is not given by~\eqref{eq:nonintersectingtorus}. } 
\be
\frac{1}{\mathcal{Z}_N}\det\left(\frac{1}{2\pi}\Theta\left(\frac{x_i-x_j}{2\pi}, \frac{it}{\pi}\right)\right)_{i,j,=1}^{2n+1},
\label{eq:nonintersectingtorus}
\ee
on the $x_i$'s, with respect to $\de x_1\cdots\de x_{2n+1}$ on $[0,2\pi)^{2n+1}$. This can be thought as a model of non-intersecting paths on the torus (non-intersecting loops).

Remarkably, the integration in~\eqref{eq:MSNunit} can be done and the joint density of the eigenvalues of $U$ turns out to be exactly the model of non-intersecting paths~\eqref{eq:nonintersectingtorus}. The diagonalisation of this matrix model is a technical matter and is postponed.

The next theorem states that the finite temperature CUE process, defined in the previous section, is the eigenvalue process of a matrix ensemble. This new matrix ensemble is nothing but the grand canonical version of~\eqref{eq:MSNunit} where the number of size of the random matrix $n$ is itself a random variable. The following result is therefore the solution of Problem~\ref{prob:3}; the proof  is an application of~\cite[Theorem 1.5 ]{Johansson07} to the formula~\eqref{eq:nonintersectingtorus}.

\begin{theorem}[Matrix model for the finite temperature CUE]\label{thm:matrixCUET} Consider $\mathrm{U}(2n+1)$ endowed with the measure~\eqref{eq:MSNunit}.
Introduce $T,\mu>0$, and denote by $N$ the integer-valued random variable defined by
\be
\Pr(N=n)=\frac{1}{\mathcal{Z}(\mu,T)}\exp\left({\frac{\mu}{T}n}\right)\frac{\mathcal{Z}_{n,\frac{2}{T}}}{n!},\quad \mathcal{Z}(\mu,T)=\sum_{n=0}^{\infty}\exp\left({\frac{\mu}{T}n}\right)\frac{\mathcal{Z}_{n,\frac{2}{T}}}{n!}.
\label{eq:Nrandom}
\ee
Consider the ensemble of random matrices $U$ of random size $N$, with law 
\be
P_{N,\frac{2}{T}}(U)\de U
\ee
(i.e., first choose $N$ according to~\eqref{eq:Nrandom} and then independently sample $U$). Then, the eigenvalues of $U$ form a determinantal point process with correlation kernel $K_{\mathrm{CUE}(T,\mu)}(x,y)$ given in~\eqref{eq:CUET}.
\end{theorem}

Before proving that the joint distribution of the eigenvalues of~\eqref{eq:MSNunit}  are given by the determinant~\eqref{eq:nonintersectingtorus} , we consider the generalisation of the construction~\eqref{eq:MSNunit} when one replaces the unitary group with a generic compact simple Lie group $G$:
\be
F_G(g_1,g_2;t)=\int_GK_t(g_1gg_2^{-1}g^{-1})\de g,
\ee
where $\de g$ is the normalised Haar measure  on $G$ ($\int_G\de g=1$) and $K_t(g)$ is the heat kernel on $G$. 
Let $\mathfrak{g}$ be the (real) Lie algebra and $\mathfrak{g}^{\C}$ its complexification. Denote by  $\mathrm{Ad}$ the adjoint representation of $G$ in $\mathfrak{g}$, and let $\langle\cdot,\cdot\rangle$ be an invariant form on  $\mathfrak{g}^{\C}$ which is positive on 
$i\mathfrak{g}$.  Denote $\mathfrak{h}$ the commutative subalgebra of $\mathfrak{g}$ of maximal dimension (the \emph{Cartan subalgebra}). Its Lie group $T=\mathrm{Lie}(\mathfrak{h})$ is the maximal torus of $G$. Choose a set of positive roots $\Sigma_{+}\subset\mathfrak{h}$ (we identify the dual by means of $\langle\cdot,\cdot\rangle$). To each positive root $\alpha\in\Sigma_{+}$ one associates the coroot $\check{\alpha}=2\alpha/\langle\alpha,\alpha\rangle$. The \emph{coroot lattice} $\check{Q}$ is the lattice generated by the coroots and it is dual to the \emph{weight lattice} $P$. Let $W$ be the Weyl group and $m_1,\dots,m_l$ the exponents of $W$ ($l=\dim\mathfrak{h}$).

The set of highest weights of the irreducible unitary representations of $G$ is $P_{+}$. Denote by $\chi_{\lambda}(g)$ and $d_{\lambda}$ the character and the dimension of the representation corresponding to $\lambda\in P_{+}$. Then 
\be
K_t(g)=\sum_{\lambda\in P_+}d_{\lambda}\chi_{\lambda}(g)e^{-c_{\lambda}t/2},
\ee
where  $c_{\lambda}=|\lambda+\rho|^2-|\rho|^2$ is the value of quadratic Casimir for the representation of weight $\lambda$, and the \emph{Weyl vector} $\rho$ is
\be
\rho =\frac{1}{2}\sum_{\alpha\in\Sigma_+}\alpha.
\ee
Observe that $F(g_1,g_2;t)$ is a central function. Introduce the\emph{Weyl denominator}
\be
\sigma(h)=e^{i\langle\rho,h\rangle}\prod_{\alpha\in\Sigma_+}(1-e^{-i\langle\alpha,h\rangle}),
\ee 
and denote $g_j=\exp(i x_j)$, $j=1,2$. Then, using Peter-Weyl theorem
\be
F_G(g_1,g_2;t)=\sum_{\lambda\in P_{+}}\chi_\lambda(e^{i(x_1-x_2)})e^{-c_{\lambda}t/2}.
\label{eq:integral_hkg}
\ee
Weyl’s formula for the characters reads
\be
\chi_\lambda(e^{ix})=\frac{\sum_{w\in W}\epsilon(w)e^{i\langle\lambda+\rho,w(x)\rangle}}{\sum_{w\in W}\epsilon(w)e^{i\langle\rho,w(x)\rangle}}=\frac{1}{\sigma(x)}\sum_{w\in W}\epsilon(w)e^{i\langle\lambda+\rho,w(x)\rangle},
\ee
where $\epsilon(w)=(-1)^{l(w)}$, $l(w)=$ length of $w$ expressed as a product of reflections. As $c_{\lambda}=|\lambda+\rho|^2-|\rho|^2$, the addenda in~\eqref{eq:integral_hkg} are quadratic in $\lambda$, and the summation over $\lambda$ may be
extended from the Weyl chamber to the full weight lattice $P$.  After some manipulations, one finds
\be
F_G(g_1,g_2;t)=\frac{e^{|\rho|^2t/2}}{\sigma(x_1)\sigma(-x_2)}\sum_{\lambda\in P}\sum_{w\in W}\epsilon(w)e^{i\langle\lambda+\rho,x_1-w(x_2)\rangle-|\lambda+\rho|^2t/2}.
\label{eq:integral_groupheatkernel}
\ee
The formula can be written as a theta function by using Poisson summation formula to convert the sum over the weight lattice into a sum over the coroot lattice (see~\cite{Altschuler91})). 

Much more could be said on the whole subject of the heat kernel on compact Lie groups. The reader will find a more substantial treatment in the large literature devoted to this subject~\cite{Berard,Bourbaki46,Kirillov08}.

We now specialise the previous formulae to the classical compact groups. 
\begin{theorem}\label{thm:MNSgroup} Let $G$ denote one of the groups $\mathrm{U}(2N+1)$, $\mathrm{SO}(2N+1)$, $\mathrm{Sp}(2N)$ and $\mathrm{SO}(2N)$ endowed with the normalized Haar measure $\de U$. Denote by $K^{G}_t(g)$ the group heat kernels and consider the random matrix $U$ with $G$-invariant law
\be
P^{G}(U)\de U=C_{t}\left(\int_{G}K_t^G(UVU^{-1}V^{-1})\de V\right)\de U.
\label{eq:meas_KtG}
\ee
Then, the joint distribution of the nontrivial eigenvalues of $U$ has density
\begin{align}
p_{\mathrm{U(2N+1)}}(x_1,\dots,x_{2N+1};t)&=\frac{1}{Z^A_{N,t}}\det\left(p_t^{\mathrm{A}}(x_i,x_j)\right)_{i,j=1}^{2N+1},\label{eq:detA}\\
p_{\mathrm{SO(2N+1)}}(x_1,\dots,x_{N};t)&=\frac{1}{Z^B_{N,t}}\det\left(p_t^{\mathrm{B}}(x_i,x_j)\right)_{i,j=1}^{N},\label{eq:detB}\\
p_{\mathrm{Sp(2N)}}(x_1,\dots,x_{N};t)&=\frac{1}{Z^C_{N,t}}\det\left(p_t^{\mathrm{C}}(x_i,x_j)\right)_{i,j=1}^{N},\label{eq:detC}\\
p_{\mathrm{SO(2N)}}(x_1,\dots,x_{N};t)&=\frac{1}{Z^D_{N,t}}\det\left(p_t^{\mathrm{D}}(x_i,x_j)\right)_{i,j=1}^{N},\label{eq:detD}
\end{align}
with respect to the Lebesgue measure on $[0,2\pi)$ in the first case, and $[0,\pi)$ otherwise.
The `kernels' are: 
\begin{align}
p_t^{\mathrm{A}}(x,y)&=\sum_{k\in\Z}e^{-k^2t/2 }e^{ik(x-y)},\label{eq:ptA}\\
p_t^{\mathrm{B}}(x,y)&=\sum_{k\in\Z+\frac{1}{2}}e^{-k^2t/2}\sin(kx)\sin(ky),\label{eq:ptB}\\
p_t^{\mathrm{C}}(x,y)&=\sum_{k\in\Z}e^{-k^2t/2}\sin(kx)\sin(ky),\label{eq:ptC}\\
p_t^{\mathrm{D}}(x,y)&=\sum_{k\in\Z}e^{-k^2t/2}\cos(kx)\cos(ky),\label{eq:ptD}
\end{align}
and $Z^A_{N,t}$,\dots,$Z^D_{N,t}$ are normalisation constants.
\end{theorem}
\begin{rmk}
The superscripts in  $p_t^{\mathrm{A}}(x,y)$, \dots, $p_t^{\mathrm{D}}(x,y)$, stand for the classical notation $A$, $B$, $C$, $D$ in the Killing-Cartan classification of semisimple Lie algebras. Once normalised, these kernels have an obvious interpretation as transition densities for a Brownian motion in an interval with periodic, absorbing or reflecting boundary conditions.
\qed
\end{rmk}

Formulae~\eqref{eq:detA}-\eqref{eq:ptD} are an application of the general result~\eqref{eq:integral_groupheatkernel} to the classical compact groups. Nevertheless, we could not find a reference that collects explicitly those formulae. For this reason we present a detailed proof.
\begin{proof}[Proof of Theorem~\ref{thm:MNSgroup}]
The proof is case by case (it cannot be any other way, as `classical groups' are defined by a list rather than by a general definition). For the reader convenience, we collect here the ingredients used in the proof. 
\begin{center}
\begin{tabular}{l|cc}
 & $\Sigma_{+}$ & $\rho$ \\ \hline\hline\\
$\mathfrak{su}(n)$  & $\{e_i-e_j\}$ & $(\frac{n-1}{2},\frac{n-3}{2},\dots,\frac{1-n}{2})$ \\ 
$\mathfrak{so}(2n+1)$  & $\{e_i\pm e_j,e_i\}$ & $(n-\frac{1}{2},n-\frac{3}{2},\dots,\frac{1}{2})$ \\ 
$\mathfrak{sp}(2n)$  & $\{e_i\pm e_j,2e_i\}$ & $(n,n-1,\dots,1)$ \\ 
$\mathfrak{so}(2n)$  & $\{e_i\pm e_j\}$ & $(n-1,n-2,\dots,0)$ \\ \\
 & $W$ & $P$ \\ \hline\hline\\
$\mathrm{U}(n)$  & $S_n$& $\Z^n$\\ 
$\mathrm{SO}(2n+1)$  & $S_n\ltimes(\Z_2)^n$& $\Z^n$\\
$\mathrm{Sp}(2n)$  & $S_n\ltimes(\Z_2)^n$& $\Z^n$\\ 
$\mathrm{SO}(2n)$  & $S_n\ltimes\{\text{even number of sign changes}\}$& $\Z^n$\\ 
\end{tabular}
\end{center}
\medskip

\par
$\mathrm{U(2N+1)}$: Consider the Lie group $G=\mathrm{U}(n)$ with Lie algebra $\mathfrak{g}=\mathfrak{u}(n)=\{iX\in\C^{n\times n}\colon X=X^{\dagger}\}$. 
The maximal torus $T$ is the subgroup of diagonal unitary matrices and the Cartan subalgebra $\mathfrak{h}$ is the algebra of diagonal matrices. The weight lattice is $P=\Z^n$, and the roots $\alpha\in\mathfrak{h}'$ of the Lie algebra are usually denoted as $\omega_{kl}=e_k-e_l$ with action
\be
\omega_{kl}(X)=x_k-x_l.
\ee
Note that $\langle\omega_{kj},\omega_{kl}\rangle=2$, so that $\rho=(\rho_1,\dots,\rho_n)=\frac{1}{2}(n-1,n-3,\dots,3-n,1-n)$. Observe that, if $n$ is even the entries of $\rho$ are half-integers; if $n$ is odd the entries of $\rho$ are integers. The Weyl denominator is the usual Vandermonde determinant $
\sigma(h)=\prod_{j<k}2i\sin\left(\frac{x_j-x_k}{2}\right)$,
and the Weyl group of $\mathrm{U}(n)$ is the group of permutations $S_n$, so that $\epsilon(w)=\mathrm{sgn}(w)$. Then, the general formula~\eqref{eq:integral_groupheatkernel} reads
\be
F_{\mathrm{U}(n}(e^{ix},e^{iy};t)=\frac{e^{|\rho|^2t/2}}{\sigma(x)\sigma(-y)}\sum_{\substack{\lambda\in (\Z+\gamma)^n}}\sum_{w\in S_n}\mathrm{sgn}(w)e^{i\langle\lambda,x-w(y)\rangle-\frac{1}{2}|\lambda|^2t},
\label{eq:kernint1}
\ee
with $\gamma=0$ if $n$ is odd, and $\gamma=1/2$ if $n$ is even. 
Setting $x=y$ and $n=2N+1$, we get
\begin{align}
F_{\mathrm{U}(2N+1)}(e^{ix},e^{ix};t)&=\frac{e^{|\rho|^2t/2}}{\displaystyle\prod_{j<k}|e^{ix_j}-e^{ix_k}|^2}\nonumber\\
&\times\displaystyle\sum_{\lambda\in \Z^{2N+1}}\sum_{w\in S_{2N+1}}\mathrm{sgn}(w)\prod_{j=1}^{2N+1}e^{i\lambda_j(x_j-x_{w(j)})-\frac{1}{2}\lambda_j^2t}\nonumber\\
&=e^{|\rho|^2t/2}\frac{\det\left(\sum_{k\in\Z}e^{-k^2t/2 }e^{ik(x_i-x_j)}\right)_{i,j=1}^{2N+1}}{\displaystyle\prod_{j<k}|e^{ix_j}-e^{ix_k}|^2}.
\end{align}
When computing the eigenvalues distribution of~\eqref{eq:meas_KtG}, the Jacobian is proportional to ${\prod_{j<k}|e^{ix_j}-e^{ix_k}|^2}$. Hence,
\be
p_{\mathrm{U(2N+1)}}(x_1,\dots,x_{2N+1};t)=\frac{1}{Z_{t}}\det\left(\sum_{k\in\Z}e^{-k^2t/2 }e^{ik(x_i-x_j)}\right)_{i,j=1}^{2N+1}.
\ee
When $n$ is even, formula~\eqref{eq:kernint1} leads to the determinant
\be
\det\left(\sum_{k\in\Z+\frac{1}{2}}e^{-k^2t/2 }e^{ik(x_i-x_j)}\right)_{i,j=1}^{n}
\ee
which, if normalised, has the interpretation of transition probability of nonintersecting Brownian motions on the circle~\cite[Proposition 1.1]{Liechty16} but does not correspond to the Karlin and McGregor formula.
\par
$\mathrm{SO(2N+1)}$: We repeat the procedure followed in the case of the unitary group step by step, pointing out only those instances that demand nontrivial modifications. 
The Lie algebra of $\mathrm{SO(2N+1)}$ is $\mathfrak{g}=\mathfrak{so}(2N+1)=\{X\in\C^{N\times N}\colon X+J^{-1}X^TJ=0 \}$, where $J$ is the symplectic matrix.
The Cartan subalgebra  of $\mathfrak{so}(2N+1)$  is $\mathfrak{h}=\mathfrak{so}(2N+1)\cap\{ \text{diagonal matrices}\}=\{\diag(x_1,\cdots,x_N,-x_1,\dots,-x_N)\}$. 
The roots $\alpha\in\mathfrak{h}'$ of the Lie algebra are $\pm e_k\pm e_l$ ($k\neq l$) and $\pm2e_i$.  The weight lattice  of $\mathrm{SO}(2N+1)$  is $(\Z+1/2)^N$ and the Weyl group is the signed symmetric group $W=S_N\ltimes(\Z_2)^N$ acting by permutations and sign changes. Simple reflections are transpositions $s_i=(i\,i+1)$ ($i=1,\dots,N-1$) and $s_N:(\lambda_1,\dots,\lambda_N)\mapsto(\lambda_1,\dots,-\lambda_N)$.
Then, \eqref{eq:integral_groupheatkernel} reads
\be
F_{\mathrm{SO}(2N+1)}(e^{ix},e^{iy};t)=\frac{e^{|\rho|^2t/2}}{\sigma(x)\sigma(-y)}\sum_{\lambda\in (\Z+\frac{1}{2})^N}\sum_{w\in S_N\ltimes(\Z_2)^N}\epsilon(w)e^{i\langle\lambda,x-w(y)\rangle-\frac{1}{2}|\lambda|^2t}.
\ee
Setting $x=y$, and `splitting' the semidirect product $S_N\ltimes(\Z_2)^N$,
\begin{align}
F_{\mathrm{SO}(2N+1)}(e^{ix},e^{ix};t)&=\frac{e^{|\rho|^2t/2}}{\displaystyle\prod_j\sin^2(x_j/2)\prod_{j<k}(2\cos{x_j}-2\cos{x_k})^2}\nonumber\\
&\times\displaystyle\sum_{\lambda\in (\Z+\frac{1}{2})^N}\sum_{\bar{w}\in S_{N}}\mathrm{sgn}(\bar{w})\prod_{j=1}^{N}(e^{i\lambda_j(x_j-x_{\bar{w}(j)})}-e^{i\lambda_j(x_j+x_{\bar{w}(j)})})e^{-\frac{1}{2}\lambda_j^2t}\nonumber\\
&=\frac{e^{|\rho|^2t/2}}{\displaystyle\prod_j\sin^2(x_j/2)\prod_{j<k}(2\cos{x_j}-2\cos{x_k})^2}\nonumber\\
&\times\det\left(\sum_{k\in\Z+\frac{1}{2}}e^{-k^2t/2 }\left(e^{ik(x_i-x_j)}-e^{ik(x_i+x_j)}\right)\right)_{i,j=1}^{N}
\end{align}
We manipulate the above expression as 
\begin{align}
&\sum_{k\in\Z+\frac{1}{2}}e^{-k^2t/2 }\left(e^{ik(x_i-x_j)}-e^{ik(x_i+x_j)}\right)\nonumber\\
=&\sum_{k\in\Z+\frac{1}{2}}e^{-k^2t/2 }\frac{1}{2}\left(e^{ik(x_i-x_j)}+e^{-ik(x_i-x_j)}-e^{ik(x_i+x_j)}-e^{-ik(x_i+x_j)}\right)\nonumber\\
=&\sum_{k\in\Z+\frac{1}{2}}e^{-k^2t/2 }\left(\cos(k(x_i-x_j))-\cos(k(x_i+x_j))\right)\nonumber\\
=&\sum_{k\in\Z+\frac{1}{2}}e^{-k^2t/2 }\;2\sin(kx_i)\sin(kx_j).\nonumber
\end{align}
Hence, 
\begin{align}
F_{\mathrm{SO}(2N+1)}(e^{ix},e^{ix};t)&=e^{|\rho|^2t/2}2^N\frac{\det\left(\sum_{k\in\Z+\frac{1}{2}}e^{-k^2t/2 }\sin(kx_i)\sin(kx_j)\right)_{i,j=1}^{N}}{\displaystyle\prod_j\sin^2(x_j/2)\prod_{j<k}(2\cos{x_j}-2\cos{x_k})^2}
\label{eq:soodd_HC}
\end{align}

Again, when computing the eigenvalue distribution, the denominator in~\eqref{eq:soodd_HC} cancels with the Jacobian (the Weyl denominator squared). 

\par
$\mathrm{Sp(2N)}$: The Lie algebra of $\mathrm{Sp(2N)}$ is $\mathfrak{g}=\mathfrak{sp}(2N)=\{X\in\C^{N\times N}\colon X+J^{-1}X^TJ=0 \}$, where $J$ is the symplectic matrix.
The weight lattice is $P=\Z^n$. 
The Cartan subalgebra  of $\mathfrak{sp}(2N)$  is $\mathfrak{h}=\mathfrak{sp}(2N)\cap\{ \text{diagonal matrices}\}=\{\diag(x_1,\cdots,x_N,-x_1,\dots,-x_N)\}$. 
The roots $\alpha\in\mathfrak{h}'$ of the Lie algebra are $\pm e_k\pm e_l$ ($k\neq l$) and $\pm2e_i$.  The weight lattice  of $\mathrm{Sp}(2N)$  is $\Z^N$ and the Weyl group is the signed symmetric group $W=S_N\ltimes(\Z_2)^N$. Steps similar to the previous case lead to 
\begin{align}
F_{\mathrm{Sp}(2N)}(e^{ix},e^{ix};t)&=e^{|\rho|^2t/2}2^N\frac{\det\left(\sum_{k\in\Z}e^{-k^2t/2 }\sin(kx_i)\sin(kx_j)\right)_{i,j=1}^{N}}{\displaystyle\prod_j\sin^2(x_j)\prod_{j<k}(2\cos{x_j}-2\cos{x_k})^2}
\label{eq:spT}
\end{align}
and then to the thesis~\eqref{eq:detC} after multiplication by the Jacobian.

\par
$\mathrm{SO(2N)}$: Specialising the general formula to this case, noting that the Weyl group contains sign changes of even parity only, we get
\begin{align}
F_{\mathrm{SO}(2N)}(e^{ix},e^{ix};t)&=\frac{e^{|\rho|^2t/2}}{\displaystyle\prod_{j<k}(2\cos{x_j}-2\cos{x_k})^2}\nonumber\\
&\times\displaystyle\sum_{\lambda\in \Z^N}\sum_{\bar{w}\in S_n}\mathrm{sgn}(\bar{w})\prod_{j=1}^{N}\left(e^{i\lambda_j(x_j-x_{\bar{w}(j)})}+e^{i\lambda_j(x_j+x_{\bar{w}(j)})}\right)e^{-\frac{1}{2}\lambda_j^2t}\nonumber\\
&=\frac{e^{|\rho|^2t/2}}{\displaystyle\prod_{j<k}(2\cos{x_j}-2\cos{x_k})^2}\nonumber\\
&\times\det\left(\sum_{k\in\Z}e^{-k^2t/2 }\left(e^{ik(x_i-x_j)}+e^{ik(x_i+x_j)}\right)\right)_{i,j=1}^{N}.
\label{eq:soeven_HC}
\end{align}
Now we use the identities
\begin{align}
&\sum_{k\in\Z}e^{-k^2t/2 }\left(e^{ik(x_i-x_j)}+e^{ik(x_i+x_j)}\right)\nonumber\\
=&\sum_{k\in\Z}e^{-k^2t/2 }\frac{1}{2}\left(e^{ik(x_i-x_j)}+e^{-ik(x_i-x_j)}+e^{ik(x_i+x_j)}+e^{-ik(x_i+x_j)}\right)\nonumber\\
=&\sum_{k\in\Z}e^{-k^2t/2 }\left(\cos(k(x_i-x_j))+\cos(k(x_i+x_j))\right)\nonumber\\
=&\sum_{k\in\Z}e^{-k^2t/2 }\;2\cos(kx_i)\cos(kx_j).\nonumber
\end{align}
to cast~\eqref{eq:soeven_HC} as 
\begin{align}
F_{\mathrm{SO}(2N)}(e^{ix},e^{ix};t)
=e^{|\rho|^2t/2}2^N\frac{\det\left(\sum_{k\in\Z}e^{-k^2t/2 }\cos(x_i)\cos(x_j)\right)_{i,j=1}^{N}}{\displaystyle\prod_{j<k}(2\cos{x_j}-2\cos{x_k})^2}
\label{eq:soeven_HCfinal}
\end{align}
\end{proof}
Generalising Theorem~\ref{thm:matrixCUET} to the finite temperature extensions of the other classical compact groups (see Section~\ref{sec:finiteTsaext}) is straightforward, after the identification $t=2/T$. Note that $t\to\infty$ corresponds to $T\to0$ and $t\to0$ to $T\to\infty$. At large time, the distribution of the non-colliding Brownian motions converges to a stationary measure, i.e. the random matrix statistics of zero temperature. At small time (large temperature) the particles behave as independent variables.

\section*{Acknowledgements}
Research of FDC and NO'C supported by ERC Advanced Grant 669306. Researh of FDC partially supported by the Italian National Group of Mathematical Physics (GNFM-INdAM). FM acknowledges support from EPSRC Grant No. EP/L010305/1.
FDC wishes to thank Paolo Facchi, Marilena Ligab\'o and Roman Schubert for very helpful discussions. 
The authors also thank an anonymous referee for many valuable comments and writing suggestions.

\end{document}